\documentclass[12pt,draftcls,onecolumn]{IEEEtran}

\usepackage{algpseudocode}
\usepackage{amsmath}
\usepackage{amssymb}
\usepackage{amsthm}
\usepackage{color}
\usepackage{cite}
\usepackage{graphicx}
\graphicspath{{./}{Figs/}}
\usepackage{subcaption}
\usepackage{mathrsfs}
\usepackage{verbatim}
\usepackage{indentfirst}
\usepackage{url}

\newtheorem{theorem}{Theorem}
\newtheorem{corollary}[theorem]{Corollary}
\newtheorem{example}[theorem]{Example}
\newtheorem{conjecture}[theorem]{Conjecture}
\newtheorem{lemma}[theorem]{Lemma}
\newtheorem{definition}[theorem]{Definition}
\newtheorem{remark}[theorem]{Remark}

\newfont{\bbb}{msbm10 scaled 500}

\newfont{\bb}{msbm10 scaled 1100}

\newcommand{\FF}{\mbox{\bb F}}
\newcommand{\GG}{\mbox{\bb G}}

\newcommand{\cv}{{\bf c}}

\newcommand{\ev}{{\bf e}}
\newcommand{\fv}{{\bf f}}
\newcommand{\gv}{{\bf g}}

\newcommand{\rv}{{\bf r}}

\newcommand{\uv}{{\bf u}}

\newcommand{\vv}{{\bf v}}
\newcommand{\xv}{{\bf x}}
\newcommand{\yv}{{\bf y}}

\newcommand{\Gm}{{\bf G}}

\newcommand{\Ac}{{\cal A}}
\newcommand{\Bc}{{\cal B}}
\newcommand{\Cc}{{\cal C}}

\newcommand{\Fc}{{\cal F}}
\newcommand{\Gc}{{\cal G}}

\newcommand{\Jc}{{\cal J}}

\newcommand{\Mc}{{\cal M}}

\newcommand{\Oc}{{\cal O}}
\newcommand{\Pc}{{\cal P}}

\newcommand{\Rc}{{\cal R}}
\newcommand{\Sc}{{\cal S}}

\newcommand{\Xc}{{\cal X}}

\newcommand{\argmax}{\arg\!\max}

\definecolor{OXO-emph}{RGB}{153,0,0}
\newcommand\htext[1]{\textbf{\textcolor{OXO-emph}{#1}}}
\newcommand\ceilb[1]{\left\lceil #1 \right\rceil}
\newcommand\floorb[1]{\left\lfloor #1 \right\rfloor}
\newcommand{\algrule}[1][.2pt]{\par\vskip.5\baselineskip\hrule height #1\par\vskip.5\baselineskip}

\title{Architecture-aware Coding for Distributed Storage: Repairable Block Failure Resilient Codes}

\author{Gokhan Calis and O.~Ozan~Koyluoglu \thanks{The authors are with Department of Electrical and Computer Engineering, The University of Arizona. Email: \{gcalis, ozan\}@email.arizona.edu. This paper was in part presented at 2014 IEEE International Symposium on Information Theory
(ISIT 2014), Honolulu, HI, June 2014.}\thanks{This work is supported in part by the National Science Foundation under Grants No CCF-1563622 and  CNS-1617335.}}

\begin{document}

\maketitle


\begin{abstract}
In large scale distributed storage systems (DSS) deployed in cloud computing, correlated failures resulting in simultaneous failure (or, unavailability) of blocks of nodes are common. In such scenarios, the stored data or a content of a failed node can only be reconstructed from the available live nodes belonging to the available blocks. To analyze the resilience of the system against such block failures, this work introduces the framework of Block Failure Resilient (BFR) codes, wherein the data (e.g., a file in DSS) can be decoded by reading out from a same number of codeword symbols (nodes) from a subset of available blocks of the underlying codeword. Further, repairable BFR codes are introduced, wherein any codeword symbol in a failed block can be repaired by contacting a subset of remaining blocks in the system. File size bounds for repairable BFR codes are derived, and the trade-off between per node storage and repair bandwidth is analyzed, and the corresponding minimum storage regenerating (BFR-MSR) and minimum bandwidth regenerating (BFR-MBR) points are derived. Explicit codes achieving the two operating points for a special case of parameters are constructed, wherein the underlying regenerating codewords are distributed to BFR codeword symbols according to combinatorial designs. Finally, BFR locally repairable codes (BFR-LRC) are introduced, an upper bound on the resilience is derived and optimal code construction are provided by a concatenation of Gabidulin and MDS codes. Repair efficiency of BFR-LRC is further studied via the use of BFR-MSR/MBR codes as local codes. Code constructions achieving optimal resilience for BFR-MSR/MBR-LRCs are provided for certain parameter regimes. Overall, this work introduces the framework of block failures along with optimal code constructions, and the study of architecture-aware coding for distributed storage systems. 

\end{abstract}


\section{Introduction}

\subsection{Background}

Increasing demand for storing and analyzing \textit{big-data} as well as several applications of cloud computing systems require efficient cloud computing infrastructures. Under today's circumstances where the data is growing exponentially, it is crucial to have storage systems that guarantee no permanent loss of the data.  However, one inevitable nature of the storage systems is node failures. In order to provide resilience against such failures, redundancy is introduced in the storage. 

Classical redundancy schemes range from \emph{replication} to \emph{erasure coding}. Erasure coding allows for better performance in terms of reliability and redundancy compared to replication. To increase repair bandwidth efficiency of erasure coded systems, regenerating codes are proposed in the seminal work of Dimakis et al. \cite{Dimakis:Network10}. In such a model of distributed storage systems (DSS), the file of size $\Mc$ is encoded to $n$ nodes such that any $k\leq n$ nodes (each with $\alpha$ symbols) allow for reconstructing the file and any $d\geq k$ nodes (with $\beta\leq \alpha$ symbols from each) reconstruct a failed node with a repair bandwidth $\gamma=d\beta$. The trade-off between per node storage ($\alpha$) and repair bandwidth ($\gamma$) is characterized and two ends of the trade-off curve are named as minimum storage regenerating (MSR) and minimum bandwidth regenerating (MBR) points \cite{Dimakis:Network10}. Several explicit codes have been proposed to achieve these operating points recently \cite{Tamo:Zigzag13,Rashmi:Optimal11,Papailiopoulos:Repair13,Dimakis:Survey11,Wu:Deterministic07,Wu:Existence10,Cadambe:Permutation11}. 

Another metric for an efficient repair is repair degree $d$ and regenerating codes necessarily have $d\geq k$. Codes with locality and locally repairable codes with regeneration properties \cite{Gopalan:Locality12,Papailiopoulos:Locally12,Oggier:Self11,Rawat:Optimal14,Prakash:Optimal12,Huang:Pyramid07,Kamath:Codes12,Kamath:Explicit13,Wang:Repair14,Rawat:Locality14} allow for a small repair degree, wherein a failed node is reconstructed via local connections. Instances of such codes are recently considered in DSS~\cite{Sathiamoorthy:XORing13,Huang:Erasure12}. Small repair degree has its benefits in terms of the implementation of DSS since a failed node requires only \textit{local} connections. In particular, when more nodes are busy for recovery operations, this in turn creates additional access cost in the network.



In large-scale distributed storage systems (such as GFS \cite{Ghemawat:Google03}),  \textit{correlated failures} are unavoidable. As analyzed in~\cite{Ford:Availability10}, these simultaneous failures of multiple nodes affect the performance of computing systems severely. The analysis in~\cite{Ford:Availability10} further shows that these correlated failures arise due to the \textit{failure domains}, e.g., nodes connected to the same power source, the same update group, or the same cluster (e.g., rack) exhibit these  \emph{structured failure bursts}. The unavailability periods are transient, and largest failure bursts almost always have significant rack-correlation. For example, in Fig.~\ref{fig:tor} there are three racks in a DSS, each connected to the same switch. Assume that top-of-rack (TOR) switch of the first rack is failed; hence, all the disks in the rack are assumed to be failed or unavailable (same failure domain). Now consider that, while the first rack is unavailable, a user asks for some data $D$ stored in one of the disks in the first rack. If both the data and its corresponding redundancy were stored all together in the first rack, then the user would not be able to access the data until the TOR switch works properly. On the other hand, assume that redundancy is distributed to the other two racks; then, the user could connect to those two racks in order to reconstruct the data. Furthermore, if the  disk storing the data fails in the first rack, then the repair process could be performed similarly since the failed disk could connect to other racks to download some amount of data for repair process. This architecture is also relevant to disk storage, where the disk is divided into sectors, each can be unavailable or under failure. To overcome from failures having such patterns, a different approach is needed.

\begin{figure}
 \centering
 \includegraphics[width=0.4\columnwidth]{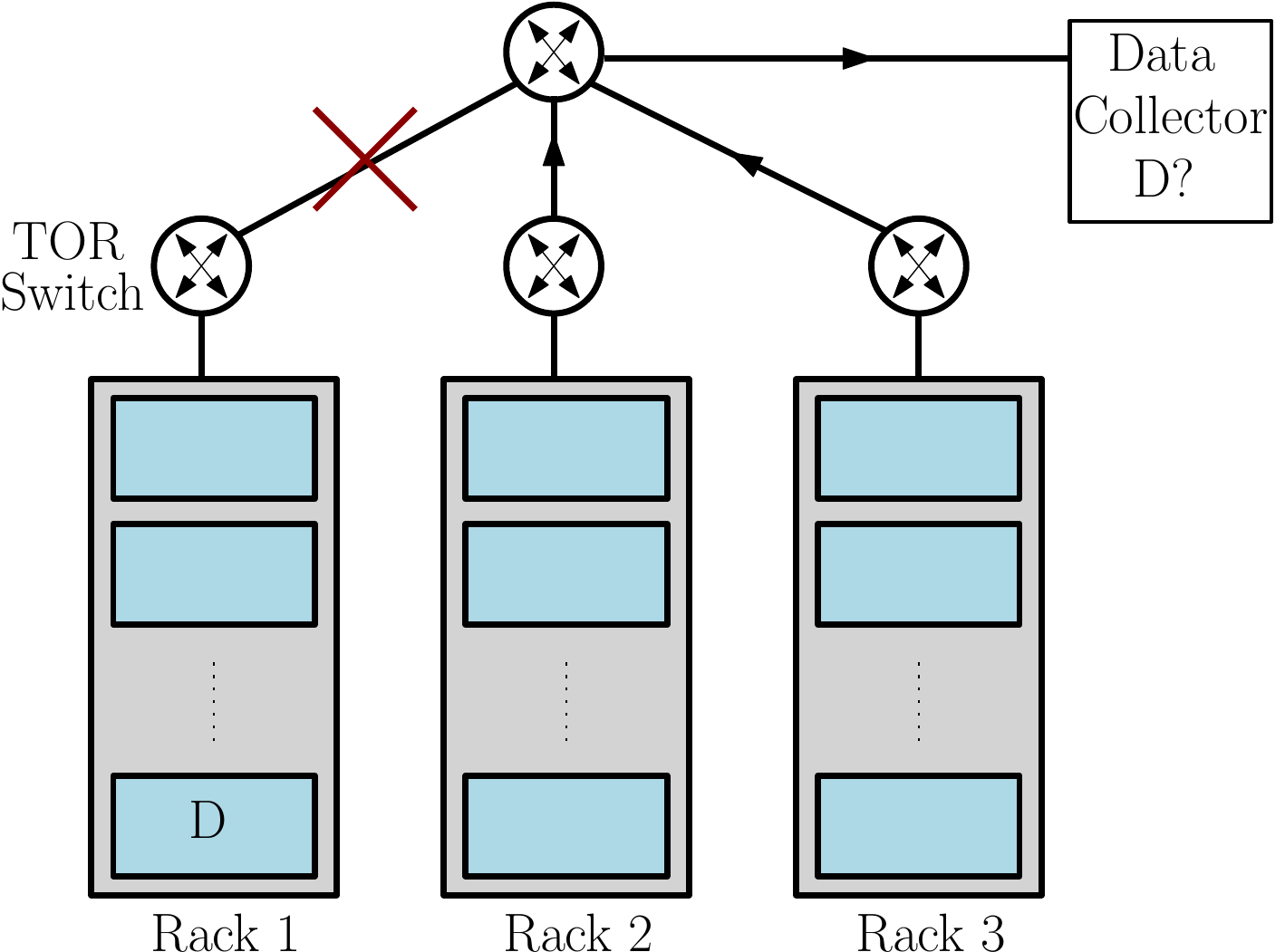}
 \caption{A data-center architecture where top-of-rack (TOR) switch of the first rack fails}
\label{fig:tor}
\vspace{-0.35in}
\end{figure}

In terms of data recovery operations, such an architecture can be considered as a relaxation of regenerating coded system. In particular, considering a load balancing property for the racks utilized in the recovery operations, we focus on the following model. Regenerating codes allow the failed node to connect to any $d$ nodes for repair purposes whereas we consider failed node connecting to a total of $d$ nodes in some restricted subset of nodes, i.e., any $\frac{d}{2}$ nodes in the second and third racks respectively for the example in Fig.~\ref{fig:tor}. Similarly, any $k$ property of regenerating codes for data-reconstruction is also relaxed, i.e., DC can connect to $\frac{k}{3}$ nodes in each rack in the example above. The outcome of such a relaxation is directly related to the performance of the DSS in terms of per node storage and repair bandwidth trade-off. For example, consider a DSS where a file of size $\Mc$ is stored over $n=10$ nodes such that any $k=4$ nodes are enough to reconstruct the data. In addition, any $d=4$ nodes are required to regenerate a failed node. For such a system that uses regenerating code, the trade-off curve can be obtained as in Fig.~\ref{fig:bfr-reg5}(c). Now consider that, $n=10$ nodes are distributed into two distinct groups such that each \textit{block} has $\frac{n}{2}=5$ nodes. Note that, a failed node in one of the blocks can now be repaired by connecting to any $d=4$ nodes in the other block. Also, DC can contact $\frac{k}{2}=2$ nodes per block to reconstruct the original message $\Mc$. For such a relaxation, we can obtain the corresponding trade-off curve between per node storage and repair bandwidth as in Fig.~\ref{fig:bfr-reg5}(c). Observe that the new MSR point, called \textit{BFR-MSR}, has significantly lower repair bandwidth than MSR point as a result of this relaxation. In this paper, we further show that the gap between the repair bandwidth of BFR-MSR and MBR point can be made arbitrarily small for large systems while keeping the per node storage of BFR-MSR point same as MSR point. Therefore, such a relaxation of regenerating codes allow for simultaneous achievability of per-node storage of MSR point and repair bandwidth of MBR point.


\begin{figure*}
	\begin{center}
		\setlength{\tabcolsep}{+0.09in}
		\begin{tabular}{ccc}
			\includegraphics[height=1.5in,width=1.6in]{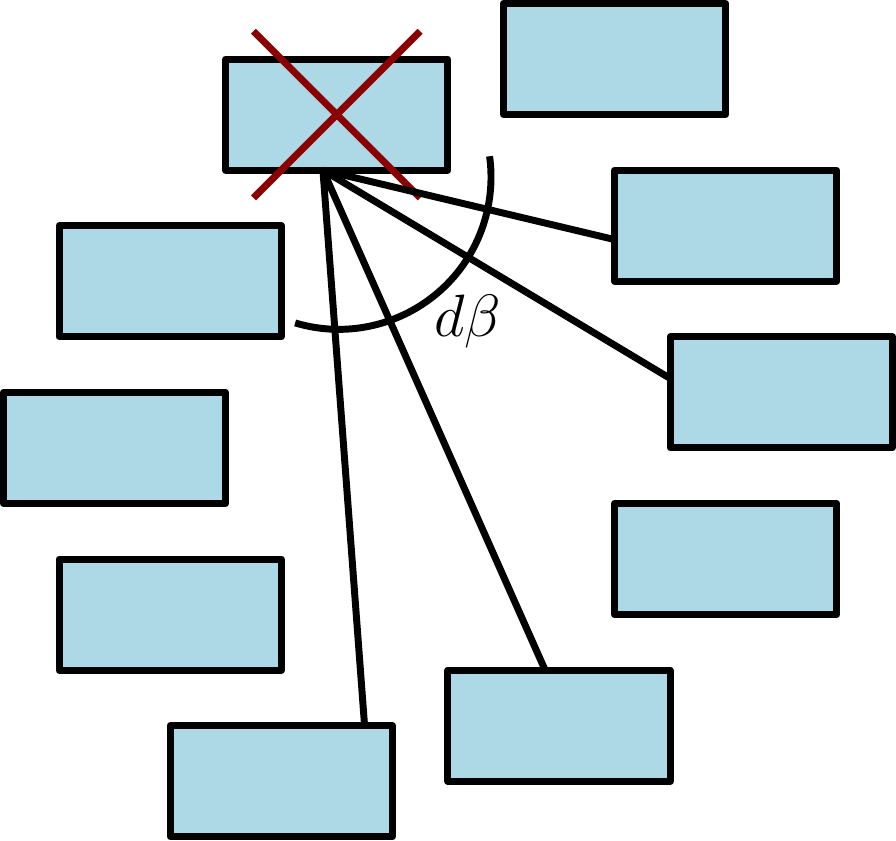} & \quad
			\includegraphics[height=1.5in,width=1.6in]{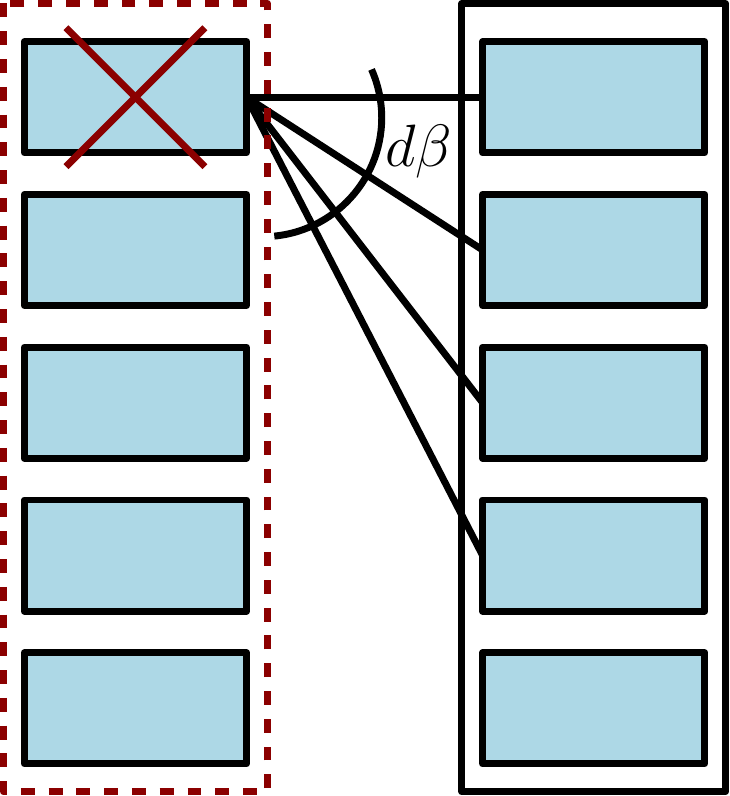} & \quad
			\includegraphics[height=1.6in,width=2.2in]{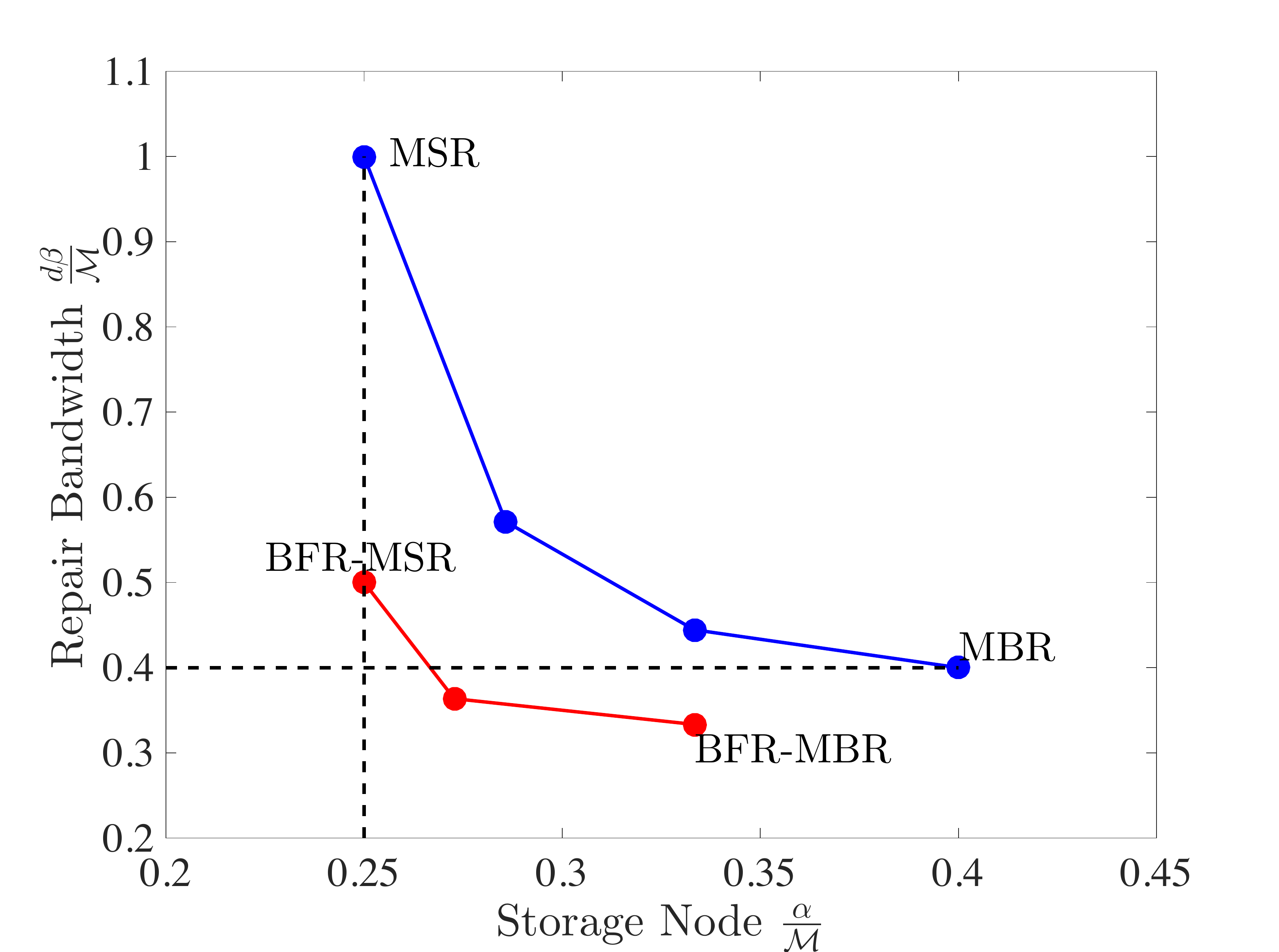}  \\
			(a) & (b) & (c) \\
		\end{tabular}
	\end{center}
	\vspace{-0.15in}
	\caption{(a) Node repair in regenerating codes. (b) Node repair in relaxation of regenerating codes. (c) Trade-off curves in toy example.}
\label{fig:bfr-reg5} 
	\vspace{-0.2in}
\end{figure*}

\subsection{Contributions and organization}

The contributions of this paper can be summarized as follows:

\begin{itemize}
\item We develop a framework to analyze resilience against block failures in DSS with node repair efficiencies. We consider a DSS with a single failure domain, where nodes belonging to the same failure group form a block of the codeword.
\item We introduce block failure resilient (BFR) codes, which allow for data collection from any $b_c= b-\rho$ blocks, where $b$ is the number of blocks in the system and $\rho$ is the resilience parameter of the code. Considering a load-balancing among blocks, a same number of nodes are contacted within these $b_c$ blocks. (A total of $k=k_cb_c$ nodes and downloading $\alpha$, i.e., all-symbols from each.) This constitutes data collection property of BFR codes. ($\rho=0$ case can be considered as a special case of batch codes introduced in ~\cite{Ishai:Batch04}.) We emphasize that our results in this part of the manuscript also consider the case of $\rho>0$. Furthermore, we introduce repairability in BFR codes, where any node of a failed block can be reconstructed from any $d_r$ nodes of any remaining $b_r=b-\sigma\leq b-1$ blocks. (A total of $d=d_rb_r$ nodes and downloading $\beta$ symbols from each.) As introduced in~\cite{Dimakis:Network10}, we utilize graph expansion of DSS employing these repairable codes, and derive file size bounds and characterize BFR-MSR and BFR-MBR points for a wide set of parameter regimes. We emphasize the relation between $d_r$ and $k_c$ as well as $\rho$ and $\sigma$ to differentiate all cases. \footnote{The case with $\frac{d}{b-\sigma} \geq \frac{k}{b-\rho}$ and $\sigma > \rho$ remains open in terms of a general proof for the minimum possible file size bound. However, the obtained bound is shown to be tight by code constructions for certain parameter regimes.} 
\item The main focus of this manuscript is on code constructions for $\sigma>0$ and $\rho=0$ scenario, where the data access assumes the availability of all $\rho$ blocks and partial data access (``repair") is over any given subset of blocks of size $b-\sigma$.
\item We construct explicit codes achieving these points for a set of parameters when $d_r \geq k_c$, $\rho=0$ and $\sigma=1$.  For a system with $b=2$ blocks case, we show that achieving both MSR and MBR properties simultaneously is asymptotically possible. (This is somewhat similar to the property of Twin codes \cite{Rashmi:Enabling11}, but here the data collection property is different.) Then, for a system with $b\geq 3$ blocks case, we consider utilizing multiple codewords, which are placed into DSS via a combinatorial design (projective planes) based codeword placement. 
\item We also provide the code constructions for any $\sigma<b-1$ and $\rho=0$ by utilizing Duplicated Combination Block Designs (DCBD). These codes can achieve BFR-MSR and BFR-MBR points for a wider set of parameters than the codes described above. 
 
\item We introduce BFR locally repairable codes (BFR-LRC). We establish an upper bound on the resilience of such codes (using techniques similar to \cite{Rawat:Optimal14}) and propose resilience-optimal constructions that use a two step encoding process and combinatorial design placement.
\item We also analyze the case where regenerating codes are used to improve repair efficiency of BFR-LRC codes. These codes are called BFR-MSR/MBR-LRC and they have a better performance in terms of repair bandwidth. We identify the upper bound on the file size that can be stored for both cases and also propose a construction that utilizes DCBD based construction to achieve the studied bound for certain parameter regimes.
\item We provide codes with table-based type of relaxation for repairable BFR in order to operate within wider set of parameters.
\end{itemize}

The rest of the paper is is organized as follows. In Section \ref{sec:Background}, we introduce BFR codes and provide preliminaries. Section \ref{sec:FileSize} focuses on file size bounds for efficient repair where we study BFR-MSR and BFR-MBR points. In Section \ref{sec:CodeConst}, we provide explicit code constructions. We discuss BFR-LRC and local regeneration in BFR-LRC in Section \ref{sec:BFR-LRC}. We extend our discussion in Section \ref{sec:Disc} where we  analyze repair time of DSS with BFR as well as other codes, e.g., regenerating codes, and also propose a relaxation for the BFR model, where we provide explicit codes achieving the performance of BFR-MSR/MBR. Finally we conclude in Section~\ref{sec:Conclusion}. 

\subsection{Related Work}

In the seminal work of Dimakis et al. \cite{Dimakis:Network10}, \textit{regenerating codes} are presented, which are proposed to improve upon classical erasure codes in terms of repair bandwidth. The authors focus on a setting similar to maximum distance separable (MDS) codes as regenerating codes also have \textit{any k out of n} property which allows data collector to connect to any $k$ nodes to decode data stored over $n$ nodes; however, they show that by allowing a failed node to connect to $d \geq k$ nodes, they can significantly reduce the repair bandwidth. Thus, a trade-off between per node storage and repair bandwidth for a single node repair is presented. Such a repair process is referred to as \textit{functional} repair since the \textit{newcomer} node may not store the same data as the failed node; however in terms of functionality, they are equivalent. In \cite{Dimakis:Network10,Wu:Deterministic07,Wu:Existence10}, functional repair is studied to construct codes achieving the optimality at the two ends of the trade-off curve. In \cite{Rashmi:Optimal11}, the focus is on \textit{exact} repair where the failed node is repaired such that the newcomer stores exactly the same data as the failed node stored. The proposed code construction is optimal for all parameters $(n,k,d)$ for minimum bandwidth regeneration (MBR) point as well as optimal for $(n,k,d\geq 2k-2)$ for minimum storage regeneration (MSR) point. It can be deducted that since $n-1 \geq d$ necessarily, these codes' rate is bounded by $\frac{1}{2}+\frac{1}{2n}$ for MSR point. \cite{Papailiopoulos:Repair13} utilizes two parity nodes to construct exact regeneration codes that are optimal in repair bandwidth, whereas \cite{Cadambe:Permutation11,Tamo:Zigzag13} focus on exact-repair for systematic node repair through permutation-matrix based constructions. Our BFR model here can be seen as a structured relaxation of the data collection and repair process of regenerating codes as discussed in the previous subsection.

Combinatorial designs are utilized in constructing fractional repetition codes in \cite{Rouayheb:Fractional10}. In this construction, MDS codewords are replicated and these replicated symbols are placed to storage nodes according to the combinatorial designs. The resulting coding scheme can be viewed as the relaxation of regenerating codes where the repairs are performed in a \emph{table-based} manner (instead of ``any $d$" property). Other works that utilize combinatorial designs include \cite{Tian:Layered14,Olmez:Fractional16,Silberstein:Optimal15}. For the constructions proposed in this paper, we consider having (in certain cases precoded versions of) \emph{regenerating codewords} placed to storage nodes according to combinatorial designs. As compared to fractional repetition codes, this allows to have bandwidth efficient repair with``any $d_r$" property in the block failure context studied here. In the last part of the sequel, we revisit and provide codes with a table-based type of relaxation for repairable BFR model, where instead of ``any $d_r$" per helper block we consider ``any $d$" per helper sub-block.

Proposed recently in \cite{Oggier:Self11}, \textit{local} codes operate node repairs within some local group thereby reducing repair degree. In \cite{Gopalan:Locality12}, minimum distance of locally repairable codes are studied for scalar case only and Pyramid codes are shown to be optimal in terms of minimum distance \cite{Huang:Pyramid07}. \cite{Prakash:Optimal12} generalizes the bound by allowing multiple local parities per local group for the scalar linear code. Next, \cite{Papailiopoulos:Locally12} focuses on the minimum distance bound for vector codes which have only one local parity per group as well as characterize the trade-off between per node storage and resilience. In \cite{Rawat:Optimal14,Kamath:Codes14},  the authors generalize the bound on minimum distance for vectors codes which can have multiple parity nodes per local group and characterize the trade-off between per node storage and resilience. \cite{Tamo:Family14} proposes locally repairable codes with small field sizes by considering polynomial construction of Reed-Solomon codes. For LRCs, when the number of failures exceed the corresponding minimum distance of the local group, the local recovery would not be possible. Our BFR model here can be considered as the opposite of locality property, where non-local nodes have to be contacted for recovery operations. 

Recently, another type of erasure codes called Partial-MDS are proposed for RAID systems to tolerate not just disk failures but also sector failures \cite{Blaum:Partial13}. Considering a stripe on a disk system as an $r \times n$ array where disks are represented by columns and all disks have $r$ sectors, PMDS codes tolerate an erasure of any $m$ sectors per row plus any $s$ additional elements. Later, relaxation of PMDS codes are studied in \cite{Blaum:Construction13} which allow erasure of any $m$ columns plus any $s$ additional elements. These codes are called Sector-Disk (SD) codes. Our BFR model can also be considered as a class of SD codes since BFR construction tolerates $\rho$ block failures which can be treated as disks (columns) in SD code construction. However, for the full access case, BFR codes do not have any additional erasure protection beyond $m=\rho$ disk failures. And, for the partial access case, BFR codes allow additional erasure protection but the code tolerates any additional $c-k_c$ number of erasures per block (disk) rather than $s$ number of erasures that can occur anywhere over the remaining disks. On the other hand, repair process in SD codes utilize the row-wise parities (one sector from other disks) to correct the corresponding sectors in the failed disks, but repairable BFR allows the failed node (sector) in a block (disk) to contact any set of nodes in other blocks.   

We note that the blocks in our model can be used to model clusters (e.g., racks) in DSS. Such a model is related to the work in \cite{Gaston:realistic13} which differentiates between within-rack communication and cross-rack communication. Our focus here would correspond to the case where within the failure group, communication cost is much higher than the cross-group communication cost, as no nodes from the failed/unavailable group can be contacted to regenerate a node. Another recent work \cite{Prakash:Storage17} considers clustered storage as well, but the repair mechanism is different than our model, where, similar to \cite{Gaston:realistic13}, \cite{Prakash:Storage17} considers availability of nodes in the block that has a failed node, i.e., there is no consideration of failure of blocks.

\section{Background and Preliminaries}
\label{sec:Background}
\subsection{Block failure resilient codes and repairability}

Consider a code $\Cc$ which maps $\Mc$ symbols (over $\FF_q$) in $\fv$ (file) to length $n$ codewords (nodes) $\cv=(c_1,\dots,c_n)$ with $c_i\in \FF_q^\alpha$ for $i=1,\dots,n$. These codewords are distributed into $b$ blocks each with block capacity $c=\frac{n}{b}$ nodes per block. We have the following definition.

\begin{definition}[Block Failure Resilient (BFR) Codes]
An $(n,b,\Mc,k,\rho,\alpha)$ block failure resilient (BFR) code encodes $\Mc$ elements in $\FF_q$ ($\fv$) to $n$ codeword symbols (each in $\FF_q^\alpha$) that are grouped into $b$ blocks such that $\fv$ can be decoded by accessing to any $\frac{k}{b-\rho}$ nodes from each of the  $b-\rho$ blocks.
\label{def:BFR}
\end{definition}

We remark that, in the above, $\rho$ represents the resilience parameter of the BFR code, i.e., the code can tolerate $\rho$ block erasures. Due to this data collection (file decoding) property of the code, we denote the number of blocks accessed as $b_c=b-\rho$ and number of nodes accessed per block as $k_c=\frac{k}{b_c}$. Noting that $k_c\leq c$ should be satisfied, we differentiate between \emph{partial} block access, $k_c<c$, and \emph{full} block access $k_c=c$. Throughout the paper, we assume $b \mid n$. i.e., $c$ is integer, and $(b-\rho) \mid k$, i.e., $k_c$ is integer.


Remarkably, any MDS array code \cite{MacWilliams:Theory77} can be utilized as BFR codes for the full access case. In fact, such an approach will be optimal in terms of minimum distance, and therefore for resilience $\rho$. However, for $k_c<c$, MDS array codes may not result in an optimal code in terms of the trade-off between resilience $\rho$ and code rate $\frac{\Mc}{n\alpha}$. Concatenation of Gabidulin codes
and MDS codes as originally proposed in \cite{Rawat:Optimal14} gives optimal BFR codes for all parameters. For completeness, we provide this coding technique adapted to generate BFR codes in Appendix \ref{app:gab_mds}. We remark that this concatenation approach is  used for locally repairable codes in\cite{Rawat:Optimal14}, for locally repairable codes with minimum bandwidth node repairs in \cite{Kamath:Explicit13}, thwarting adversarial errors in \cite{Silberstein:Error12,Silberstein:Error15}, cooperative regenerating codes with built-in security mechanisms against node capture attacks in \cite{Koyluoglu:Secure14} and for constructing PMDS codes in \cite{Calis:General16}. In this work, we focus on repairable BFR codes, as defined in the following. 

\begin{definition}[Block Failure Resilient Regenerating Codes (BFR-RC)]
An $(n,b,\Mc,k,\rho,\alpha,$ $d,\sigma,\beta)$ block failure resilient regenerating code (BFR-RC) is an $(n,b,\Mc,k,\rho,\alpha)$ BFR code (data collection property) with the following repair property: Any node of a failed block can be reconstructed by accessing to any $d_r=\frac{d}{b-\sigma}$ nodes of any $b_r=b-\sigma$ blocks and downloading $\beta$ symbols from each of these $d=b_rd_r$ nodes.
\end{definition}

We assume $(b-\sigma)\mid d$, i.e., $d_r$ is integer. (Note that $d_r$ should necessarily satisfy $\frac{d}{b-\sigma}=d_r\leq c=\frac{n}{b}$ in our model.) We consider the trade-off between the \emph{repair bandwidth} $\gamma=d\beta$ and \emph{per node storage} $\alpha$ similar to the seminal work of Dimakis et al. \cite{Dimakis:Network10}. In particular, we define $\alpha_{\textrm{BFR-MSR}}=\frac{\Mc}{k}$ as the minimum per node storage and $\gamma_{\textrm{BFR-MBR}}=\alpha_{\textrm{BFR-MBR}}$ as the minimum repair bandwidth for an $(n,b,\Mc,k,\rho,\alpha,d,\sigma,\beta)$ BFR-RC.

\begin{remark}
	Having $\sigma>0$ allows partial data reconstruction as one can regenerate the unavaliable/failed node inside a block.
\end{remark}

\subsection{Information flow graph}

The operation of a DSS employing such codes can be modeled by a multicasting scenario over an information flow graph \cite{Dimakis:Network10}, which has three types of nodes: 1) Source node ($S$): Contains original file $\fv$. 2) Storage nodes, each represented as $x_i$ with two sub-nodes($(x^{\rm in}_i,x^{\rm out}_i)$), where $x^{\rm in}$ is the sub-node having the connections from the live nodes, and $x^{\rm out}$ is the storage sub-node, which stores the data and is contacted for node repair or data collection (edges between each $x^{\rm in}_i$ and $x^{\rm out}_i$ has $\alpha$-link capacity). 3) Data collector ($\rm{DC}$) which contacts $x^{\rm{out}}$ sub-nodes of $k$ live nodes (with edges each having $\infty$-link capacity). (As described above, for BFR codes these $k$ nodes can be any $\frac{k}{b-\rho}$ nodes from each of the $b-\rho$ blocks.) Then, for a given graph $\Gc$ and DCs $\rm{DC}_i$, the file size can be bounded using the max flow-min cut theorem for multicasting utilized in network coding~\cite{Ho:random06,Dimakis:Network10}.
\begin{lemma}[Max flow-min cut theorem for multicasting] \label{lma:MFMCforMulticast}
$$\Mc \leq \min_{\Gc} \min_{\rm{DC}_i} \rm{max flow}(S \to \rm{DC}_i,\Gc),$$
where $\rm{flow}(S \to \rm{DC}_i,\Gc)$ represents the flow from the source
node $S$ to $\rm{DC}_i$ over the graph $\Gc$.
\end{lemma}
Therefore, $\mathcal{M}$ symbol long file can be delivered to a DC, only if the min cut is at least $\mathcal{M}$. In the next section, similar to \cite{Dimakis:Network10}, we consider $k$ successive node failures and evaluate the min-cut over possible graphs, and obtain file size bounds for DSS operating with BFR-RC.

\subsection{Vector codes}

\looseness=-1 An $(n,M,d_{\rm{min}},\alpha)_q$ vector code $\Cc \subseteq (\FF_q^\alpha)^n$ is a collection of $M$ vectors of length $n\alpha$ over $\FF_q$. A codeword $\cv \in \Cc$ consists of $n$ blocks, each of size $\alpha$ over $\FF_q$. We can replace each $\alpha$-long block with an element in $\FF_{q^\alpha}$ to obtain a vector $\cv = (\cv_1,\cv_2,\dots,\cv_n) \in \FF_{q^\alpha}^n$. The minimum distance, $d_{\rm{min}}$, of $\Cc$ is defined as minimum Hamming distance between any two codewords in $\Cc$.

\begin{definition}
Let $\cv$ be a codeword in $\Cc$ selected uniformly at random from $M$ codewords. The minimum distance of $\Cc$ is defined as $$d_{\rm{min}}=n-\max\limits_{\Ac \subseteq[n]:H(\cv_\Ac) < \log_q M}|\Ac|,$$
where $\Ac=\left\{i_1,\dots,i_{|\Ac|}\right\}\subseteq [n]$, $\cv_\Ac=(\cv_{i_1},\dots,\cv_{i_{|\Ac|}})$, and $H(\cdot)$ denotes $q$-entropy.
\end{definition}

A vector code is said to be maximum distance separable (MDS) code if $\alpha \mid \log_q M$ and $d_{\rm{min}}=n-\frac{log_q M}{\alpha}+1$. A linear $(n,M,d_{\rm{min}},\alpha)_q$ vector code is a linear subspace of $\FF_q^{\alpha n}$ of dimension $\Mc=\log_q M$. An $[n,\Mc,d_{\rm{min}},\alpha]_q$ array code is called \emph{MDS array code} if $\alpha \mid \Mc$ and $d_{min} = n- \frac{\Mc}{\alpha}+1$.

The encoding process of an $(n,M=q^\Mc,d_{\rm{min}},\alpha)_q$ vector code can be summarized by 
$\GG:\FF_q^\Mc \rightarrow \left(\FF_q^\alpha \right)^n$. The encoding function is defined by an $\Mc \times n\alpha$ generator matrix $G=[\gv_1^1,\dots,\gv_1^\alpha \mid \dots \mid \gv_n^1,\dots,\gv_n^\alpha]$ over $\FF_q$. 

\subsection{Maximum rank distance codes}

Gabidulin codes \cite{Gabidulin:Theory85}, are an example of class of rank-metric codes, called maximum rank distance (MRD) codes. (These codes will be utilized later in the sequel.)

Let $\FF_{q^m}$ be an extension field of $\FF_q$. An element $\textit{v} \in \FF_{q^m}$ can be represented as the vector $\textbf{\textit{v}} = (\textit{v}_1,\dots,\textit{v}_m)^T \in \FF_{q}^{m}$, such that $\textit{v}=\sum_{i=1}^{m} v_i b_i$, for a fixed basis $\left\{b_1,\dots,b_m\right\}$ of the extension field $\FF_{q^m}$. Using this, a vector $\vv = (v_1,\dots,v_N) \in \FF_{q^m}^N$ can be represented by an $m \times N$ matrix $\textbf{V}=[v_{i,j}]$ over $\FF_q$, which is obtained by replacing each $v_i$ of $\vv$ by its vector representation $(v_{i,1},\dots,v_{i,m})^T$.

\begin{definition}
The rank of a vector $\vv \in \FF_{q^m}^N$, rank($\vv$), is defined as the rank of its $m \times N$ matrix representation $\textbf{V}$ (over $\FF_q$). Then, rank distance is given by $$d_R(\vv , \uv)=rank(\textbf{V}-\textbf{U}).$$ 
\end{definition}  

An $[N,K,D]_{q^m}$ rank-metric code $C \subseteq \FF_{q^m}^{N}$ is a linear block code over $\FF_{q^m}$ of length $N$ with dimension $K$ and minimum rank distance $D$. A rank-metric code that attains the Singleton bound $D \leq N-K+1$ in rank-metric is called \emph{maximum rank distance} (MRD) code. Gabidulin codes can be described by evaluation of \emph{linearized polynomials}.

\begin{definition}
A linearized polynomial $f(y)$ over $\FF_{q^m}$ of $q$-degree $t$ has the form $$f(y)=\sum_{i=0}^{t}a_i y^{q^i}$$ where $a_i \in \FF_{q^m}$, and $a_t \neq 0$.
\end{definition}
Process of encoding a message $(f_1,f_2,\dots,f_K)$ to a codeword of an $[N,K,D=N-K+1]$ Gabidulin code over $\FF_{q^m}$ has two steps:

\begin{itemize}
\item Step 1: Construct a data polynomial $f(y)=\sum_{i=1}^K f_i y^{q^{i-1}}$ over $\FF_{q^m}$.
\item Step 2: Evaluate $f(y)$ at $\left\{y_1,y_2,\dots,y_n\right\}$ where each $y_i \in \FF_{q^m}$, to obtain a codeword $\cv = (f(y_1),\dots,f(y_N)) \in F_{q^m}^N$.
\end{itemize}

\begin{remark}
For any $a,b \in \FF_q$ and $\textit{v}_1 , \textit{v}_2 \in \FF_{q^{m}}$, we have $f(a \textit{v}_1+b \textit{v}_2)=a f(\textit{v}_1)+b f(\textit{v}_2)$.
\end{remark}

\begin{remark}
Given evaluations of $f(\cdot)$ at any $K$ linearly independent (over $\FF_{q}$) points in $\FF_{q^{m}}$, one can reconstruct the message vector. Therefore, an $[N,K,D]$ Gabidulin code is an MDS code and can correct any $D-1=N-K$ erasures. 
\end{remark}

\subsection{Locally repairable codes}

Recently introduced locally repairable codes reduce repair degree by recovering a symbol via contacting small number of helper nodes for repair. 

\begin{definition}[Punctured Vector Codes]
Given an $(n,M,d_{\rm{min}},\alpha)_q$ vector code $\Cc$ and a set  $\Sc \subset [n]$, $\Cc|_\Sc$ is used to denote the code obtained by puncturing $\Cc$ on $[n]\backslash\Sc$. In other words, codewords of $C|_\Sc$ are obtained by keeping those vector symbols in $\cv=(\cv_1,\dots,\cv_n)\in \Cc$ which have their indices in set $\Sc$. 
\end{definition}

\begin{definition}
\label{def:locality}
An $(n,M,d_{\rm{min}},\alpha)_q$ vector code $\Cc$ is said to have $(r,\delta)$ \emph{locality} if for each vector symbol $\cv_i \in \FF_q^n$, $i \in [n]$, of codeword $\cv = (\cv_1,\dots,\cv_n) \in \Cc$, there exists a set of indices $\Gamma(i)$ such that
\begin{itemize}
\item $i \in \Gamma(i)$
\item $|\Gamma(i)|\leq r+\delta-1$
\item Minimum distance of $\Cc|_{\Gamma(i)}$ is at least $\delta$.
\end{itemize}
\end{definition}

\begin{remark}
The last requirement in Definition \ref{def:locality} implies that each element $j \in \Gamma(i)$ can be written as a function of any set of $r$ elements in $\Gamma(i)\backslash \left\{j\right\}$. 
\end{remark}

\section{File Size Bound for Repairable BFR Codes}
\label{sec:FileSize}

In this section, we perform an analysis to obtain file size bounds for repairable BFR codes. We focus on different cases in order to cover all possible values of parameters $d_r$, $k_c$, $\rho$ and $\sigma$. 

We denote the whole set of blocks in a DSS by $\Bc$, where $|\Bc|=b$. A failed node in block $i$ can be recovered from any $b-\sigma$ blocks. Denoting an instance of such blocks by $\Bc_i^r$, we have $|\Bc_i^r|=b-\sigma$ and $i \notin \Bc_i^r$. This repair process requires a total number of $d$ nodes and equal number of nodes from each block in $\Bc_i^r$, $d_r = \frac{d}{b-\sigma}$, is contacted. Data collector, on the other hand, connects to $b-\rho$ blocks represented by $\Bc^{c}$ to reconstruct the stored data. In this process, a total number of $k$ nodes are connected to retrieve the stored data and the same number of nodes from each block in $\Bc^{c}$, $k_c=\frac{k}{b-\rho}$, are connected.

In the following, we analyze the information flow graph for a total number of $k_c(b-\rho)$ failures and repairs, and characterize the min-cut for the corresponding graphs to upper bound the file size that can be stored in DSS. 

Our analysis here differs from the classical setup  considered in \cite{Dimakis:Network10}, in the sense that the scenarios considered here require analysis of different cases (different min-cuts in information flow graph) depending on the relation between $d_r$ and $k_c$ ($d_r \geq k_c$ vs. $d_r < k_c$). Before detailing the analysis for each case (in Section \ref{subsec:d_r>=k_c} and \ref{subsec:d_r<k_c}), we here point out why this phenomenon occurs. First, consider the case $d_r \geq k_c$, a failed node in block $i$ requires connections from the blocks $j \in \Bc^r_i$ and that are previously failed and repaired and connected to DC. At this point, block $i$ needs to contact some additional nodes in block $j$ ($d_r-k_c$ number of nodes) that may not be connected to DC. On the other hand, when $d_r < k_c$, then, for the min-cut, any failed node in block $i$ would contact $d_r$ nodes in block $j$, which are already connected to DC. Therefore, the storage systems exhibits different min-cuts depending on the values of $d_r$, $k_c$, $\rho$ and $\sigma$. In the following, we analyze each case separately, obtain the min-cuts as well as corresponding corner points in per-node storage and repair bandwidth trade-off. Section \ref{subsec:d_r>=k_c} focuses on $d_r \geq k_c$ case, where the number of helper nodes from each block is greater than or equal to the number of nodes connected to DC from each block. Section \ref{subsec:d_r<k_c} details the remaining case, i.e., $d_r<k_c$.

\subsection{Case I: $d_r\geq k_c$}
\label{subsec:d_r>=k_c}


\subsubsection{Case I.A: $\sigma \leq \rho$}

The first case we focus on is having $\sigma \leq \rho$, which implies that the number of blocks participating to the repair process is greater than or equal to the number of blocks that are connected to DC, i.e., $b-\sigma \geq b-\rho$. 


\begin{theorem}
If $\alpha=\frac{\Mc}{k}$ or $\alpha=d \beta$, the upper bound on the file size when $d_r\geq k_c$ and $\sigma \leq \rho$ is given by
\begin{equation}
\Mc  \leq  \sum_{i=1}^{b-\rho}k_c\min \left\{\alpha,(d-(i-1)k_c)\beta \right\}.
\label{eq:boundcase1}
\end{equation}
\label{thm:case1}
\vspace{-0.25in}
\end{theorem}
\begin{proof}
Denote the order of node repairs with corresponding (node) indices in the ordered set $\Oc$, where $|\Oc|=k_c(b-\rho)$. We will show that any order of failures results in the same cut value. We have $\mathscr{C}_{\Oc} = \sum_{i=1}^{b-\rho}k_c\min \left\{\alpha,(d-(i-1)k_c\beta) \right\}$. This follows by considering that block $i$ connects to $i-1$ blocks, each having $k_c$ nodes that are repaired and connected to DC, which makes the cut value of $(d-(i-1)k_c)\beta$. We denote the stored data of node $j$ in block $l$ as $\Xc_{l,j}$ and downloaded data to this node as $\Rc_{l,j}$. The set $\Oc$ includes an ordered set of node incides $(i,j)$ contacted by DC, where the order indicates the sequence of failed and repaired nodes. (For instance, $\Oc = \left\{(1,1),(1,2),(2,3),\cdots \right\}$ refers to a scenario in which third node of second blocks is repaired after second node of first block.) We consider DC connecting to the nodes in $$\Oc= \left\{ (1,1),(1,2),\cdots,(1,k_c),(2,1),(2,2),\cdots,(2,k_c),\cdots,(b-\rho,1),\cdots,(b-\rho,k_c) \right\},$$ where each block has repairs sequentially. Due to data reconstruction property, it follows that $H(\Fc|\Xc_{\Oc})=0$. Accordingly, we have $$\Mc=H(\Fc)=H(\Fc)-H(\Fc|\Xc_{\Oc})=I(F;\Xc_{\Oc})\leq H(\Xc_{\Oc}).$$ We consider the following bound on the latter term $H(\Xc_{\Oc})$ to obtain $\mathscr{C}_{\Oc}$, i.e., \textit{``cut value"} for the repair order given by $\Oc$. (Note that the bounds given below correspond to introducing \textit{``cuts"} in the corresponding information flow graph.) We denote $\Oc(i) \triangleq \left\{ (i,1),(i,2),\cdots,(i,k_c) \right\} $ as the nodes contacted in $i$-th block. We consider each repaired node contacts the previously repaired nodes in $\Oc$. Accordingly, we consider the entropy calculation given by 
\begin{eqnarray}
    H(\Xc_{\Oc}) &=& \sum_{i=1}^{b-\rho} H(\Xc_{\Oc(i)}|\Xc_{\Oc(1)},\cdots,\Xc_{\Oc_{(i-1)}}) 
    \leq \sum_{i=1}^{b-\rho} \sum_{(i,j)\in \Oc(i)} H(\Xc_{(i,j)}|\Xc_{\Oc(1)},\cdots,\Xc_{\Oc_{(i-1)}})   \\
    &\stackrel{(a)}{\leq}& \sum_{i=1}^{b-\rho} \sum_{(i,j)\in \Oc(i)} \min \left\{\alpha,(d-(i-1)k_c)\beta \right\}   \stackrel{(b)}{=} \sum_{i=1}^{b-\rho}k_c\min \left\{\alpha,(d-(i-1)k_c)\beta \right\} \triangleq \mathscr{C}_{\Oc} \nonumber
\label{eq:entropy_case1-condt_ent}
\end{eqnarray}
where (a) follows as $H(\Xc_{(i,j)}|\Xc_{\Oc(1)},\cdots,\Xc_{\Oc(i-1)})\leq H(\Xc_{(i,j)})\leq \alpha$ and $$H(\Xc_{(i,j)}|\Xc_{\Oc(1)},\cdots,\Xc_{\Oc(i-1)}) \leq H(\Rc_{(i,j)}|\Xc_{\Oc(1)},\cdots,\Xc_{\Oc(i-1)}) \leq (d-(i-1)k_c)\beta,$$ as $i-1$ blocks containing previously repaired nodes contributes to $(i-1)k_c\beta$ symbols in $\Rc_{(i,j)}$. Therefore, as the bound on $H(X_{(i,j)})$ above holds for any $j$ such that $(i,j)\in \Oc_{(i)}$, we have (b), from which we obtain the following bound on file size $\Mc$ $$\Mc  \leq \sum_{i=1}^{b-\rho}k_c\min \left\{\alpha,(d-(i-1)k_c)\beta \right\}.$$ Interchanging the failure order in $\Oc$ at indices $k_c$ and $k_c+1$, we obtain another order $\Oc^*$. Using the same bounding technique above, the resulting cut value can be calculated as $\mathscr{C}_{\Oc^*} = (k_c-1)\min \left\{\alpha,d\beta \right\} + \min \left\{\alpha,(d-(k_c-1))\beta \right\} + \min \left\{\alpha,(d-1)\beta \right\} + (k_c-1)\min \left\{\alpha,(d-k_c)\beta \right\} + \sum_{i=3}^{b-\rho}k_c\min \left\{\alpha,(d-(i-1)k_c)\beta \right\}$. Note that the first term here corresponds to the first set of $k_c-1$ nodes in the first blocks, second term corresponds to the first node repaired in the second block, third term corresponds to the remaining node in the first block, fourth terms corresponds to remaining $k_c-1$ nodes in the second block, and the last term corresponds to the remaining blocks. We observe that $\mathscr{C}_{\Oc}=\mathscr{C}_{\Oc^*}$, i.e., the total cut values induced by $\Oc$ and $\Oc^*$ are the same. Consider $\mathscr{C}_{\Oc^*}-\mathscr{C}_{\Oc}$, which evaluates to $$\mathscr{C}_{\Oc^*}-\mathscr{C}_{\Oc}= \min \left\{\alpha,(d-1)\beta \right\} - \min \left\{\alpha,d\beta \right\} + \min \left\{\alpha,(d-(k_c-1))\beta \right\} - \min \left\{\alpha,(d-k_c)\beta \right\}.$$ We observe that $\mathscr{C}_{\Oc^*}-\mathscr{C}_{\Oc}=0$ at both $\alpha=\frac{\Mc}{k}$ and $\alpha=d\beta$ operating points. (For $\alpha=\frac{\Mc}{k}$, all $\min\left\{ \right\}$ terms result in $\alpha$ and cancel each other. For the $\alpha=d\beta$ point, $\min\left\{ \right\}$ terms remove $\alpha$ (as $\alpha=d\beta$) and remaining terms cancel each other.) 


The equivalence analysis above can be extended to any pair of orders. In the following, we show this through the following lemmas. We first observe that any order can be obtained from another by swapping adjacent failures.

\begin{lemma}
Any given failure order $\Oc^*$ can be obtained by permuting the elements in the order $\Oc$, and the underlying permutation $\pi_{\Oc \to \Oc^*}$ operation can be decomposed into stages of swapping of adjacent elements.
\label{thm:lemma_permut}
\end{lemma}

We provide an example here, the lemma above generalizes this argument to any pair of orders. Consider $\Oc=(1,2,3)$ and $\Oc^*=(3,2,1)$, the permutation $\pi_{\Oc \to \Oc^*}$ is given by position mappings $\left\{1\to3,2\to2,3\to1\right\}$. This can be obtained by composition of three permutations $\left\{1\to1,2\to3,3\to2\right\}$, for swapping 2 and 3, $\left\{1\to2,2\to1,3\to3\right\}$, for swapping 1 and 3, and $\left\{1\to1, 2\to3,3\to2\right\}$, for swapping 1 and 2. 

Utilizing Lemma~\ref{thm:lemma_permut}, it remains to show that the cut value remains the same if we interchange adjacent failures in any failure order $\Oc$. We show that this holds in the following lemma.

\begin{lemma} 
For any order $\Oc$ swapping any adjacent failures does not result in min-cut value $\mathscr{C}_\Oc$ to change. 
\end{lemma}

To show this, consider two orders $\Oc^{1}$ and $\Oc^{2}$, where $\Oc^{2}$ is obtained by swapping order of failures at locations $j$ and $j+1$, $\Oc^2=\pi(\Oc^1)$. Then, we can say that the cut values up to $j-1$ and the cut values after $j+1$ are individually the same for both $\Oc^{1}$ and $\Oc^{2}$. Furthermore, there are two possible cases for swapped failures $j$ and $j+1$; i) either both are from the same block, ii) they are from different blocks. For the former case, the swapping does not affect the cut values for failures $j$ and $j+1$. In the latter, first note that when $\alpha=\frac{\Mc}{k}$, there is no change in cut value, $\mathscr{C}_{\Oc^1}=\mathscr{C}_{\Oc^2}$, hence we'll only focus on $\alpha=d\beta$ case. Assume for the $\Oc^{1}$ we have $(d-i)\beta$ for $i^{\rm{th}}$ and $(d-i^*)\beta$ for $(i+1)^{\rm{th}}$ failures. Then, $\Oc^{2}$ should have $(d-(i^*-1))\beta$ and $(d-(i+1))\beta$ respectively. Note that the sums are still the same, $(d-i)\beta+(d-i^*)\beta=(d-(i^*-1))\beta+(d-(i+1))\beta$, hence we can conclude that swapping any two adjacent failures does not change the min-cut value, $\mathscr{C}_{\Oc^1}=\mathscr{C}_{\Oc^2}$.   

Combining two lemmas, we conclude that every order of failures has the same cut value as $\Oc$, which is $\mathscr{C}_{\Oc}$.
\end{proof}


The key distinction to be made in our scenario is that all of the nodes that are failed and repaired are utilized in the repair process of the subsequent failures, which is why the order of failures does not matter. Note that for the special case of $k_c=1$, the above bound reduces to the classical bound given in \cite{Dimakis:Network10}, albeit only for $\alpha=\frac{\Mc}{k}$ or $\alpha=d\beta$ case. We obtain the following corner points in the trade-off region.

\begin{corollary}
For $d_r \geq k_c$ and $\sigma \leq \rho$, corresponding BFR-MSR and BFR-MBR points can be found as follows.
\begin{equation}
(\alpha_{\textrm{BFR-MSR}},\gamma_{\textrm{BFR-MSR}}) = \left(\frac{\Mc}{k},\frac{\Mc d}{kd-\frac{k^2(b-\rho-1)}{b-\rho}}\right),
\label{eq:case1_MSR}
\end{equation}
\begin{equation}
(\alpha_{\textrm{BFR-MBR}},\gamma_{\textrm{BFR-MBR}}) = \left(\frac{\Mc d}{kd-\frac{k^2(b-\rho-1)}{2(b-\rho)}},\frac{\Mc d}{kd-\frac{k^2(b-\rho-1)}{2(b-\rho)}}\right).
\label{eq:case1_MBR}
\end{equation}

\label{cor:case1}
\end{corollary}
 
\begin{proof}
For BFR-MSR point, we set $\alpha=\frac{\Mc}{k}$ in \eqref{eq:boundcase1} and obtain the requirements $[d-(b-\rho-1)k_c\beta \geq \alpha]$ from which we obtain that minimum $\gamma$ occurs at $\gamma_{\textrm{BFR-MSR}}=\frac{\alpha}{d-(b-\rho-1)k_c}d$ as $\beta$ is lower bounded by $\frac{\alpha}{d-(b-\rho-1)k_c}$. BFR-MBR point, on the other hand, follows from the bound $\Mc \leq \sum_{i=1}^{b-\rho} k_c[d-(i-1)k_c]\beta$ as $\alpha=d\beta$ at the minimum bandwidth. More specifically, minimum $\beta$ is given by $\beta=\frac{\Mc}{\sum_{i=1}^{b-\rho} k_c[d-(i-1)k_c]}$ and $\gamma_{\textrm{BFR-MBR}}=\alpha_{\textrm{BFR-MSR}}=d\beta$.
\end{proof}  

\begin{remark}
When $b=\rho+1$, we observe that $\alpha_{\textrm{BFR-MSR}}=\gamma_{\textrm{BFR-MSR}}=\alpha_{\textrm{BFR-MBR}}=\gamma_{\textrm{BFR-MBR}}=\frac{\Mc}{k}$.
\end{remark}

\subsubsection{Case I.B: $\sigma > \rho$}

For the case of having $\sigma > \rho$, we have $|\Bc^{c}|=b-\rho > |\Bc^r|=b-\sigma$. Noting that the helper nodes are chosen from the ones that are already connected to DC, we consider that $\Bc^{c} \supset \Bc^r$. Here, the min-cut analysis similar to the one given in the proof of Lemma \ref{thm:case2} is same as having a system with $\bar{b}=b-\rho,\bar{\rho}=0,\bar{\sigma}=\sigma-\rho$ since the analysis for the file size bound only utilizes $b-\rho$ blocks and repairs occur by connecting to a subset of these $b-\rho$ blocks. That is, the remaining $\rho$ blocks do not contribute to the cut value. Therefore, we conclude that $\mathscr{C}(b,\rho,\sigma)=\mathscr{C}(b-\rho,0,\sigma-\rho)$ when $\sigma>\rho$. 



\begin{lemma}
An upper bound on the file size when $d_r\geq k_c$ and $\sigma > \rho$ is given by
\begin{align}
\Mc  \leq  & \sum_{i=1}^{b-\sigma}k_c\min \left\{\alpha,\beta(d-(i-1)k_c) \right\} + \sum_{i=b-\sigma+1}^{b-\rho} k_c\min \left\{\alpha,\beta(d-(b-\sigma)k_c) \right\}.
\label{eq:boundcase2}
\end{align}
\label{thm:case2}
\vspace{-0.25in}
\end{lemma}
\begin{proof}
Let $\Oc$ be an order such that first $k_c$ failures occur in the first block, the next $k_c$ failures occur in the second block and so on. Denote by $\mathscr{C}_{\Oc}$ the total cut value induced by $\Oc$. (The analysis detailed in the proof of Theorem \ref{thm:case1} is followed here.) Consider the node indexed by $i=k_c(b-\sigma)$ so that order $\Oc$ can be split into two parts as $\Oc^{i-}$ and $\Oc^{i+}$, where $\Oc^{i-}$ represents the failures up to (and including) the node $i$, and $\Oc^{i+}$ represents the remaining set of failures. 

Since any node failure contacts $b-\sigma$ blocks and noting that $\Oc^{i-}$ includes exactly $b-\sigma$ blocks, any node failure in $\Oc^{i+}$ would contribute to the cut value as $\min\left\{ \alpha, (d-(b-\sigma)k_c)\beta \right\}$. On the other hand, $\mathscr{C}_{\Oc^{i-}}$ follows from Theorem \ref{thm:case1}. Combining $\mathscr{C}_{\Oc^{i-}}$ and $\mathscr{C}_{\Oc^{i+}}$, we get \eqref{eq:boundcase2}.
\end{proof}

\begin{remark}
We conjecture that the order given in the proof above corresponds to the order producing the min-cut. We verified this with numerical analysis for systems having small $n$ values. Although a general proof of this conjecture is not established yet \footnote{The main difficulty for this case is having $\sigma > \rho$, which makes some failed nodes not being utilized in the repair process. That is why the proposed order is constructed such a way that we try to maximize the use of failed nodes in the repair process of subsequent nodes.}, we were able to construct codes achieving the stated bound. Therefore we conjecture the following MSR/MBR points for this case.
\end{remark}

Utilizing the bound given in Lemma \ref{thm:case2}, we obtain following result (proof is similar to  Corollary \ref{cor:case1} and omitted for brevity).

\begin{conjecture}
For $d_r \geq k_c$ and $\sigma > \rho$, corresponding BFR-MSR and BFR-MBR points can be found as follows.
\begin{equation}
(\alpha_{\textrm{BFR-MSR}},\gamma_{\textrm{BFR-MSR}}) = \left(\frac{\Mc}{k},\frac{\Mc d}{kd-\frac{k^2(b-\sigma)}{b-\rho}}\right),
\label{eq:case2_MSR}
\end{equation}

\begin{equation}
(\alpha_{\textrm{BFR-MBR}},\gamma_{\textrm{BFR-MBR}}) = \left(\frac{\Mc d}{kd-\frac{k^2(b-\sigma)(b+\sigma-2\rho-1)}{2(b-\rho)^2}},\frac{\Mc d}{kd-\frac{k^2(b-\sigma)(b+\sigma-2\rho-1)}{2(b-\rho)^2}}\right).
\label{eq:case2_MBR}
\end{equation}
\label{cor:case2}
\end{conjecture}

If code constructions achieve the points above (\eqref{eq:case2_MSR} and \eqref{eq:case2_MBR}), then this will result in an optimal BFR-MSR/BFR-MBR. Later, we propose code constructions achieving these points, see Section~\ref{subsec:dcbd}.

\subsection{Case II: $d_r < k_c$}
\label{subsec:d_r<k_c}
We first note that $d\geq k$. (This follows similarly to the analysis for regenerating codes, as otherwise one can regenerate each node by contacting $d<k$ nodes, and obtain the stored data with less than $k$ nodes.) Assume $\rho \leq \sigma$, then $b-\sigma \leq b-\rho$. This, together with $d_r < k_c$ implies that repair operation contacts to less number of blocks and less number of nodes per block as compared to data access. If this is the case, via a repair operation, all the remaining nodes can be regenerated and data access can be completed using repair operation. Therefore, the valid scenario is $\rho > \sigma$. (This is similar to the reasoning of $d \geq k$ assumption in regenerating codes.)



\begin{theorem}
The optimal file size when $d_r < k_c$ and $\rho > \sigma$ is given by
\begin{equation}
\Mc  \leq  \sum_{i=1}^{b-\rho}d_r\min \left\{\alpha,(d-(i-1)d_r)\beta \right\} +  \sum_{i=1}^{b-\rho} (k_c-d_r) \min \left\{\alpha,(d-(b-\rho-1)d_r)\beta \right\}.
\label{eq:boundcase3}
\end{equation}
\label{thm:case3}
\vspace{-0.2in}
\end{theorem}
\begin{proof}
Let $\Oc$ be an order and let index $i=d_r(b-\rho)$ so that order $\Oc$ can be split into two parts as $\Oc^{i-}$ and $\Oc^{i+}$ where $\Oc^{i-}$ represents the failures up to (and including) index $i$ and $\Oc^{i+}$ represents the remaining set of failures. For index $i$, $\mathscr{C}_{\Oc^{i-}}$ takes its minimum possible value if $\Oc^{i-}$ contains exactly $d_r$ failures from each of $b-\rho$ blocks. We show this by a contradiction. Assume that $\Oc^{i-}$ contains $d_r$ failures from each of $b-\rho$ blocks but corresponding $\mathscr{C}_{\Oc^{i-}}$ is not the minimum. We have already shown that the failure order among the nodes in $\Oc^{i-}$ does not matter as long as the list contains $d_r$ failures from $b-\rho$ blocks. This follows from $d_r = k_c$ case analyzed in Theorem \ref{thm:case1}. Assume an order spanning $b-\rho+t$ blocks for $t>0$. This ordering will include nodes that are not connected to DC and can be omitted. This means that an order that minimizes the cut value needs to contain a block that has at least $d_r+1$ failures. Denote such an ordering by $\Oc^{i--}$,  assume without loss of generality that this block $j$ has $t \geq 1$ additional failures. In such a case, we observe that the cut value for failed nodes  in other blocks would not be affected by this change since each such node is already connected to $d_r$ nodes of block $j$. Furthermore, by removing a failed node of some other block (other than $j$) from $\Oc^{i-}$, the cut values corresponding to other nodes in the list would only increase since the failures can benefit at most $d_r-1$ nodes of that block as opposed to $d_r$ in $\Oc^{i-}$. Hence, in order to minimize $\mathscr{C}_{\Oc^{i-}}$, $\Oc^{i-}$ needs to include exactly $d_r$ failures from each of $b-\rho$ blocks. Also, it can be observed that, when $\Oc^{i-}$ is constructed as above, then $\Oc^{i+}$ also takes its minimum possible value, since any failure in $\Oc^{i+}$ utilizes maximum possible number of repaired nodes ($d_r$ nodes from each of the $b-\rho-1$ blocks) that are connected to DC hence it only needs to contact $\rho-\sigma+1$ blocks that are not connected to DC. Therefore, both $\mathscr{C}_{\Oc^{i-}}$ and $\mathscr{C}_{\Oc^{i+}}$ are minimized individually with $\Oc^{i-}$. Therefore, for a fixed threshold $i=d_r(b-\rho)$, we obtain the min-cut. 
\end{proof}

\begin{corollary}
For $d_r < k_c $, corresponding BFR-MSR and BFR-MBR points can be found as follows.

\begin{equation}
(\alpha_{\textrm{BFR-MSR}},\gamma_{\textrm{BFR-MSR}}) = \left(\frac{\Mc}{k},\frac{\Mc d}{\frac{kd(\rho-\sigma+1)}{b-\sigma}}\right),
\label{eq:case3_MSR}
\end{equation}

\begin{equation}
(\alpha_{\textrm{BFR-MBR}},\gamma_{\textrm{BFR-MBR}}) = \left(\frac{\Mc d}{\frac{kd(\rho-\sigma+1)}{b-\sigma} + \frac{d^2(b-\rho)(b-\rho-1)}{2(b-\sigma)^2}},\frac{\Mc d}{\frac{kd(\rho-\sigma+1)}{b-\sigma} + \frac{d^2(b-\rho)(b-\rho-1)}{2(b-\sigma)^2}}\right).
\label{eq:case3_MBR}
\end{equation}
\label{cor:case3}
\end{corollary}

\subsection{BFR-MSR and BFR-MBR Points for Special Cases}

The general case is  analyzed in the previous section, and we here focus on special cases for BFR-MSR and BFR-MBR points. The first case analyzed below has the property that corresponding BFR-MSR codes achieve both per-node storage point of MSR codes and repair bandwidth of MBR codes simultaneously when $2d \gg k$. 


\subsubsection{Special Case for I.B: $\rho=0$, $\sigma=1$, $d_r \geq k_c$ and $b=2$}

\begin{figure}
 \centering
 \includegraphics[width=0.34\columnwidth]{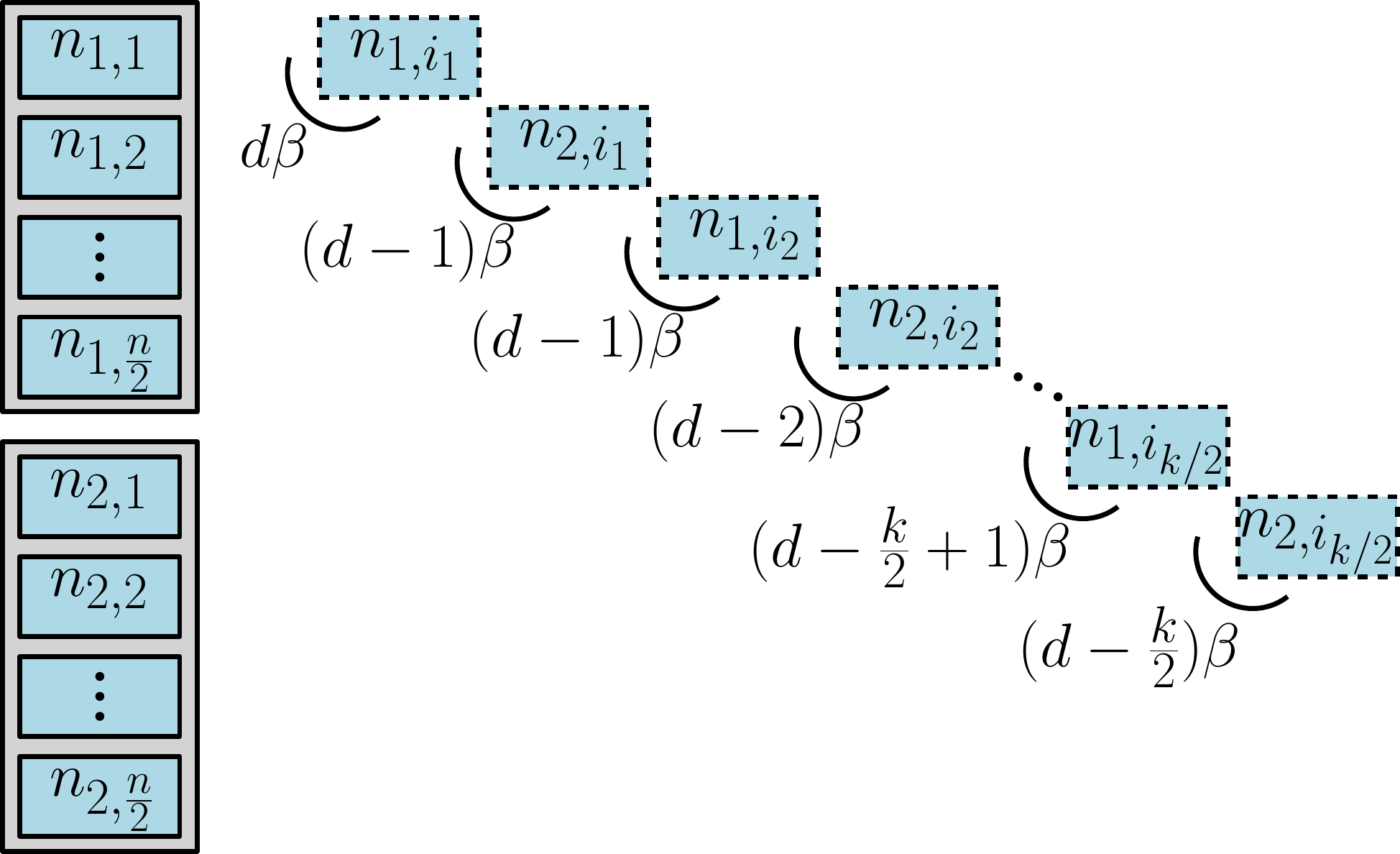}
 \caption{Repair process for $b=2$ (two blocks) case.}
\label{fig:Two-block}
\vspace{-0.3in}
\end{figure}

Consider the 2-block case $(b=2)$ as in Fig.~\ref{fig:Two-block}, and assume $2 \mid k$. The file size $\Mc$ can be upper bounded with the repair procedure shown in Fig.~\ref{fig:Two-block}, which displays one of the ``minimum-cut'' scenarios, wherein any two consecutive node failures belong to different blocks. Assuming $d\geq \frac{k}{2}$, we obtain
\begin{equation}
\Mc \leq \sum_{i=0}^{\frac{k}{2}-1}\min \{\alpha,(d-i)\beta\} + \sum_{i=1}^{\frac{k}{2}}\min\{\alpha,(d-i)\beta\}.
\label{eq:min-cut_two}
\end{equation}

 Achieving this upper bound \eqref{eq:min-cut_two} with equality would yield maximum possible file size. One particular repair instance is shown in Fig.~\ref{fig:Two-block}, and we note that the order of node repairs does not matter as the sum of the corresponding cut values would be the same with different order of failures as long as we consider connection from data collector to $\frac{k}{2}$ repaired nodes from each block. 

For BFR-MSR point, $\alpha=\alpha_{\textrm{BFR-MSR}}=\frac{\Mc}{k}$. In the bound \eqref{eq:min-cut_two}, we then have $\alpha_{\textrm{BFR-MSR}} \leq (d-\frac{k}{2})\beta_{\textrm{BFR-MSR}}$. Achieving equality would give the minimum repair bandwidth for the MSR case. Hence, BFR-MSR point is given by
\begin{equation}
(\alpha_{\textrm{BFR-MSR}},\gamma_{\textrm{BFR-MSR}}) = \left( \frac{\Mc}{k},\frac{2\Mc d}{2kd-k^2} \right).
\label{eq:min-cut_two-MSR-values}
\end{equation}
Note that, this coincides with that of \eqref{eq:case2_MSR} where we set $b=2$, $\rho=0$ and $\sigma=1$ therein.

BFR-MBR codes, on the other hand, have the property that $d\beta=\alpha$ with minimum possible $d\beta$ while achieving the equality in \eqref{eq:min-cut_two}. Inserting $d\beta=\alpha$ in \eqref{eq:min-cut_two}, we obtain that
\begin{equation}
(\alpha_{\textrm{BFR-MBR}},\gamma_{\textrm{BFR-MBR}}) = \left(\frac{4\Mc d}{4dk-k^2},\frac{4\Mc d}{4dk-k^2}\right).
\label{eq:min-cut_two-MBR-values}
\end{equation} 
This coincides with that of \eqref{eq:case2_MBR} where we set $b=2$, $\rho=0$ and $\sigma=1$ therein.

We now consider the case where $2 \nmid k$  (as compared to previous section where we assumed $k_c = \frac{k}{b-\rho}$), and characterize trade-off points for all possible system parameters in this special case. First consider the special case of $k=3$ and two different order of failures, one with first failure in first block, second failure in second block, third failure in first block and the other one with first and second failures from first block, third failure from second block. Accordingly, observe that the cuts as  $\min\{\alpha,d\beta\}+2\min\{\alpha,(d-1)\beta\}$ and $2\min\{\alpha,d\beta\}+\min\{\alpha,(d-2)\beta\}$ respectively. For MSR case, first sum would require $\alpha=(d-1)\beta$, whereas second sum requires $\alpha=(d-2)\beta$, resulting in higher repair bandwidth. Henceforth, one needs to be careful even though cut values are same for both orders of failures in  both MSR ($3\alpha$) and MBR ($(3d-2)\beta$) cases. 
\begin{equation}
\Mc \leq \sum_{i=1}^{\floorb{\frac{k}{2}}+1}\min\{\alpha,d\beta\}  + \sum_{i=1}^{\floorb{\frac{k}{2}}}\min\{\alpha,(d-\floorb{\frac{k}{2}}-1)\beta\} 
\end{equation}

 
The corresponding trade-off point are summarized below by following analysis similar to the one above.
\begin{equation}
(\alpha_{\textrm{BFR-MSR}},\gamma_{\textrm{BFR-MSR}}) =
\begin{cases}
	(\frac{M}{k},\frac{2Md}{2kd-k^2-k}), & \textrm{ if $k$ is odd} \\
    (\frac{M}{k},\frac{2Md}{2kd-k^2}), &\textrm{ o.w.}
\end{cases}
\label{eq:MSR_cases}
\end{equation}

\begin{equation}
(\alpha_{\textrm{BFR-MBR}},\gamma_{\textrm{BFR-MBR}}) =
\begin{cases}
	(\frac{4Md}{4dk-k^2+1},\frac{4Md}{4dk-k^2+1}), & \textrm{ if $k$ is odd} \\
    (\frac{4Md}{4dk-k^2},\frac{4Md}{4dk-k^2}), & \textrm{ o.w.}
\end{cases}
\label{eq:MBR_cases}
\end{equation}

Here, we compare $\gamma_{\textrm{BFR-MSR}}$ and $\gamma_{\textrm{MBR}}$. We have $\min\{\gamma^{\textrm{k-odd}}_{\textrm{BFR-MSR}}, \gamma^{\textrm{k-even}}_{\textrm{BFR-MSR}}\} \geq \gamma_{\textrm{MBR}}=\frac{2\Mc d}{k(2d-k+1)}$, and, if we have $2d-k \gg 1$, then $\gamma^{\textrm{k-odd}}_{\textrm{BFR-MSR}} \approx \gamma^{\textrm{k-even}}_{\textrm{BFR-MSR}} \approx \gamma_{\textrm{MBR}}$. This implies that BFR-MSR codes with $b=2$ achieves repair bandwidth of MBR and per-node storage of MSR codes simultaneously for systems with $d \gg 1$. On Fig.~\ref{fig:odd_regular} and \ref{fig:even_regular}, we depict the ratio of $\gamma^{\textrm{k-odd}}_{\textrm{BFR-MSR}}$ and $\gamma^{\textrm{k-even}}_{\textrm{BFR-MSR}}$ to $\gamma_{\textrm{MBR}}$ respectively, where we keep $k$ constant and vary $d$ as $2k \geq d \geq k$. Also, only even $d$ values are shown in both figures. It can be observed that ratio gets closer to 1 as we increase $k$. Next, we provide the generalization of critical points to $b \geq 2$ case in the following.

\begin{figure*}[h]
	\vspace{-0.2in}
    \centering
    \begin{subfigure}[t]{0.5\textwidth}
        \centering
 		\includegraphics[width=0.75\columnwidth]{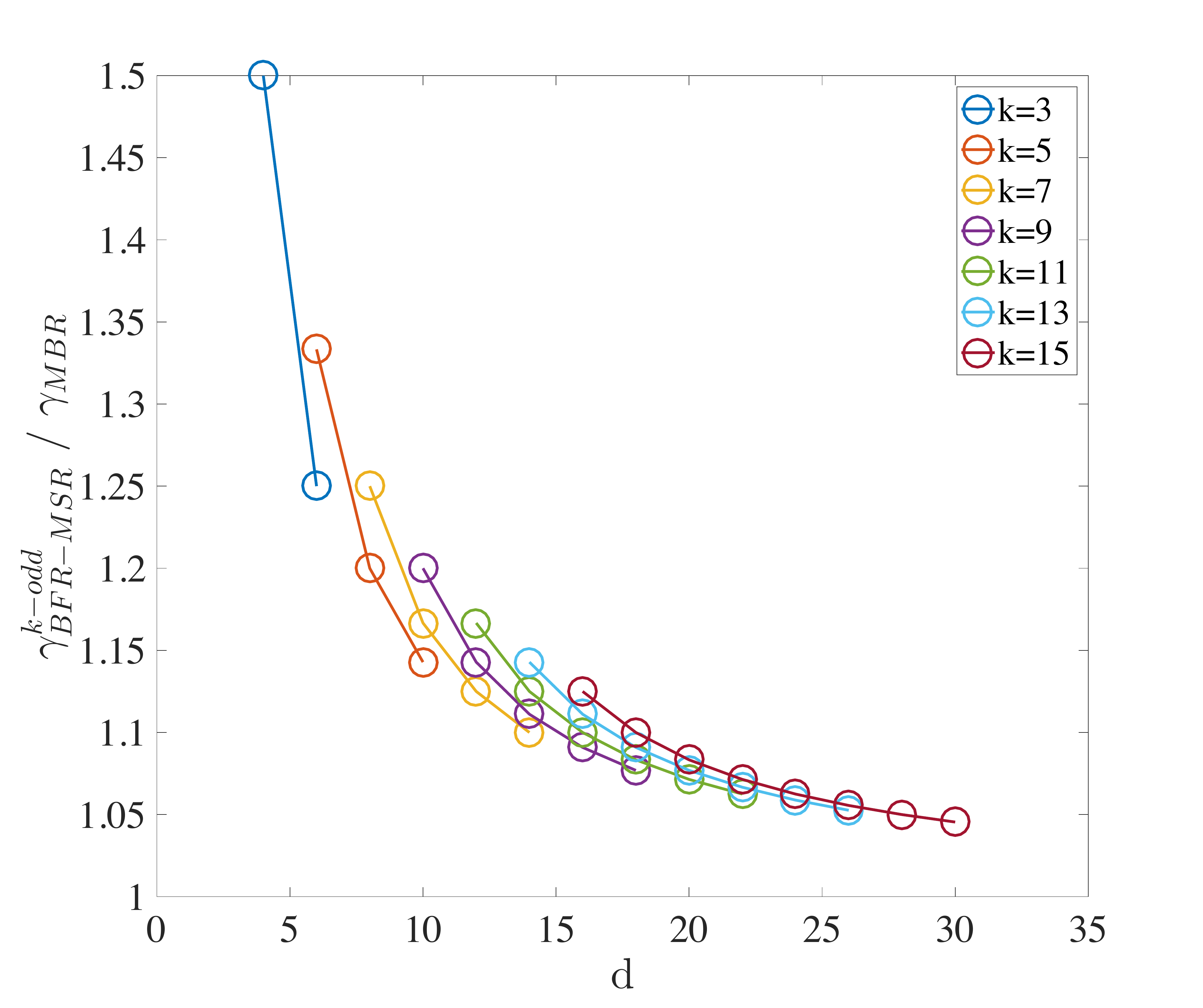}
 		\caption{}
 		\label{fig:odd_regular}
    \end{subfigure}%
    ~ 
    \begin{subfigure}[t]{0.5\textwidth}
        \centering
 		\includegraphics[width=0.75\columnwidth]{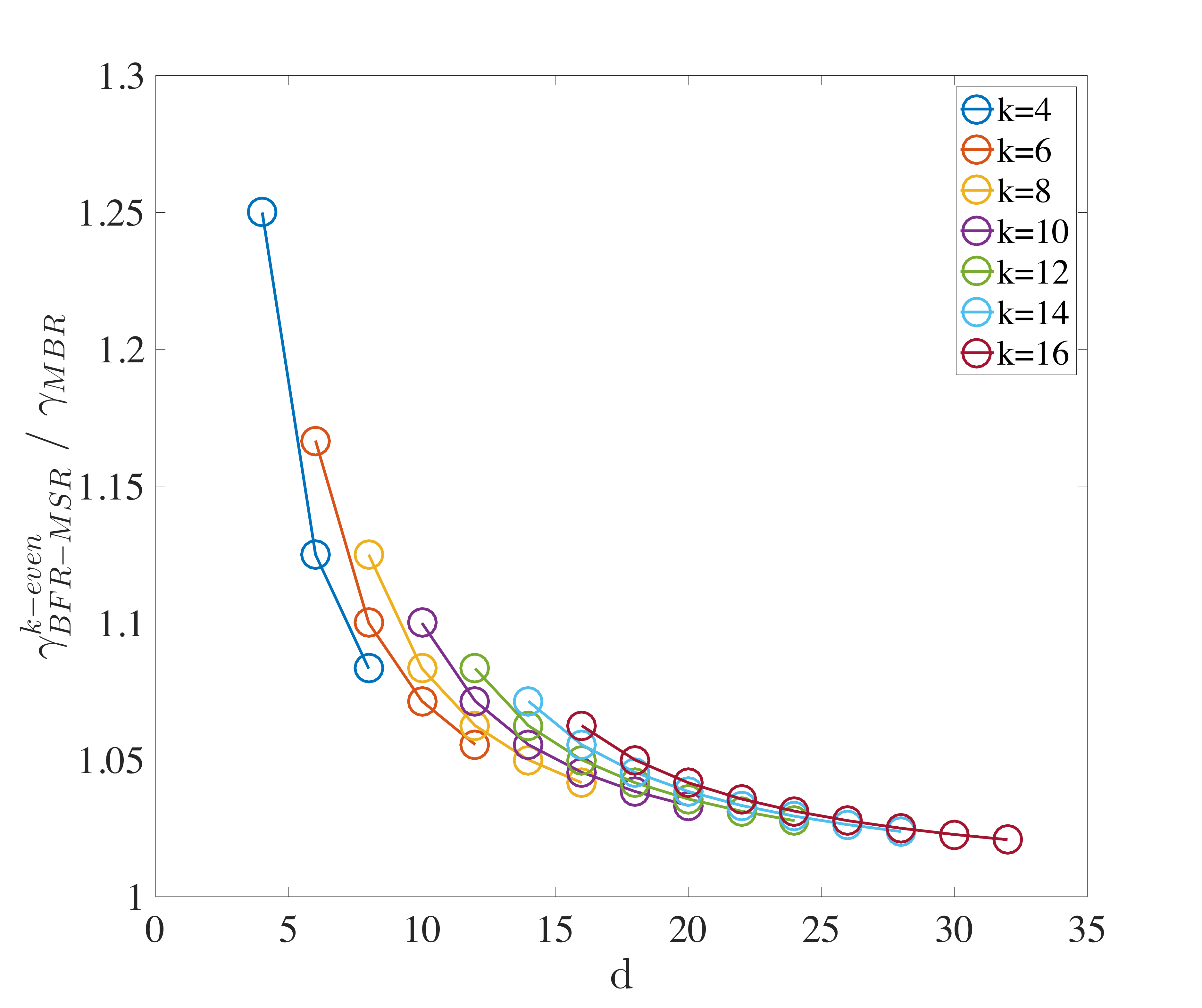}
 		\caption{}
 		\label{fig:even_regular}
    \end{subfigure}
    \vspace{-0.15in}
    \caption{(a) Ratio $\frac{\gamma^{\textrm{k-odd}}_{\textrm{BFR-MSR}}}{\gamma_{\textrm{MBR}}}$ vs. $d$. (b) Ratio $\frac{\gamma^{\textrm{k-even}}_{\textrm{BFR-MSR}}}{\gamma_{\textrm{MBR}}}$ vs. $d$.} 
    
\end{figure*}

\begin{remark}
	The bound in \eqref{eq:min-cut_two} holds for the intermediate points as well for the special case discussed in this section $(b=2,\rho=0,\sigma=1)$ (as compared to previous section where the general case is analyzed). See Fig.~\ref{fig:bfr-reg5}(c) and \cite{Calis:Repairable14}.
\end{remark}

\subsubsection{Special Case for Case I.B: $\rho=0$, $\sigma=1$, $d_r \geq k_c$ and $b\geq2$}

From Corollary \ref{cor:case2}, we obtain the corresponding MSR and MBR points in this special case.



\begin{corollary}
For $\rho=0$, $\sigma=1$, $d_r \geq k_c$ and $b\geq2$ BFR-MSR and BFR-MBR points are as follows. 
\begin{equation}
(\alpha_{\textrm{BFR-MSR}},\gamma_{\textrm{BFR-MSR}}) = \left(\frac{\Mc}{k},\frac{\Mc d}{kd-\frac{k^2(b-1)}{b}}\right)
\label{eq:min-cut_b-MSR-values}
\end{equation}

\begin{equation}
(\alpha_{\textrm{BFR-MBR}},\gamma_{\textrm{BFR-MBR}}) = \left(\frac{\Mc d}{kd-\frac{k^2(b-1)}{2b}},\frac{\Mc d}{kd-\frac{k^2(b-1)}{2b}}\right)
\label{eq:min-cut_b-MBR-values}
\end{equation}
\end{corollary}

We observe that $\gamma_{\textrm{BFR-MSR}} \leq \gamma_{\textrm{MSR}}=\frac{\Mc d}{k(d-k+1)}$ for $b \leq k$, which is the case here as we assume $b \mid k$. Also, we have $\frac{\gamma_{\textrm{BFR-MSR}}}{\gamma_{\textrm{MBR}}} = \frac{d-\frac{k-1}{2}}{d-k\frac{b-1}{b}} \geq 1$ when $b \geq \frac{2k}{k+1}$ which is always true. Hence, $\gamma_{\textrm{BFR-MSR}}$ is between $\gamma_{\textrm{MSR}}$ and $\gamma_{\textrm{MBR}}$, see Fig.~\ref{fig:bfr-reg5}(c).


\section{BFR-MSR and BFR-MBR Code Constructions}
\label{sec:CodeConst}
\subsection{Transpose code: b=2}

\textbf{\textit{Construction I (Transpose code):}} Consider $\alpha=d=\frac{n}{2}$, and placement of $\frac{n\alpha}{2}$ symbols denoted by $\left\{ x_{i,j}: i,j \in \left\{ 1,2, \cdots, \frac{n}{2}=\alpha \right\} \right\}$ for $b=2$ blocks according to the following rule: Node $i$ in (in block $b=1$) stores symbols $\left\{ x_{i,j}: j \in \left\{ 1,2, \cdots, \alpha \right\} \right\}$, whereas node $i + \frac{n}{2}$ (in block $b=2$) stores symbols $\left\{ x_{j,i}: j \in \left\{ 1,2, \cdots, \alpha \right\} \right\}$ for $i=1,2, \cdots, \frac{n}{2}$. Note that, when the stored symbols in nodes of block 1 is represented as a matrix, the symbols in block 2 corresponds to transpose of that matrix. (We therefore refer to this code as transpose code.)

Due to this transpose property, the repair of a failed node $i$ in the first block can be performed by connecting to all the nodes in the second block and downloading only $1$ symbol from each node. That is, we have $d\beta=\alpha$. Consider now that the file size $\Mc=kd-(\frac{k}{2})^2$, and an $[N=\alpha^2,K=\Mc]$ MDS code is used to encode file $\fv$ into symbols denoted with $x_{i,j}$, $i,j=1,\dots,\alpha$. Here, BFR data collection property for reconstructing the file is satisfied, as connecting to any $k_c=\frac{k}{2}$ nodes from each block assures at least $K$ distinct symbols. This can be shown as follows: Consider $\alpha \times \alpha$ matrix $X$, where $i$-th row, $j$-th column has the element $x_{i,j}$. Rows of $X$ correspond to nodes of block 1, and columns of $X$ correspond to nodes of block 2. Any $\frac{k}{2}$ rows (or any $\frac{k}{2}$ columns) provide total of $\frac{k\alpha}{2}$ symbols. And $\frac{k}{2}$ rows and $\frac{k}{2}$ columns intersect at $\left(\frac{k}{2}\right)$ symbols. Therefore, total number of symbols from any $\frac{k}{2}$ rows and $\frac{k}{2}$ columns is $\Mc$. Note that the remaining system parameters are  $d_r=\frac{n}{2} \geq k_c$, $\rho=0$ and $\sigma=1$. Henceforth, this code is a BFR-MBR code as the operating point in \eqref{eq:min-cut_b-MBR-values}, is achieved with $\Mc=kd-(\frac{k}{2})^2$ and $d\beta=\alpha$ for $\beta=1$ (scalar code). 

A similar code to this construction is Twin codes introduced in \cite{Rashmi:Enabling11}, where the nodes are split into two types and a failed node of a a given type is regenerated by connecting to nodes only in the other type. During construction of Twin codes, the message is first transposed and two different codes are applied to both original message and it's transposed version separately to obtain code symbols. On the other hand, we apply one code to the message and transpose resulting symbols during placement. Also, Twin codes, as opposed to our model, do not have balanced node connection for data collection. In particular, DC connects to only (a subset of $k$ nodes from) a single type and repairs are conducted from $k$ nodes. On the other hand, BFR codes, for $b=2$ case, connects to $\frac{k}{2}$ nodes from each block and repairs are from any $d$ nodes in other block. 

This construction, however, is limited to $b=2$ and in the following section we propose utilization of block designs to construct BFR codes for $b>2$.


\begin{figure}
 \centering
 \includegraphics[width=0.4\columnwidth]{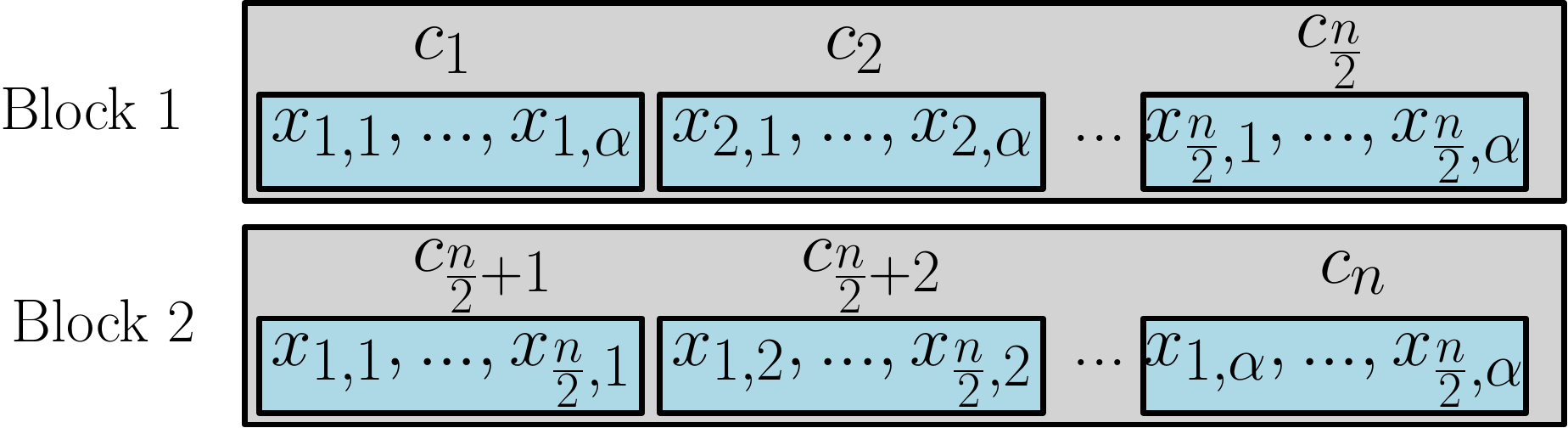}
 \caption{Transpose code is a two-block BFR-MBR code.}
\label{fig:Transpose}
\vspace{-0.25in}
\end{figure}

\subsection{Projective plane based placement of regenerating codewords ($\rho=0$, $\sigma=1$)}

\begin{figure*}[h]
    \centering
    \begin{subfigure}[t]{0.4\textwidth}
        \centering
 		\includegraphics[width=0.15\columnwidth]{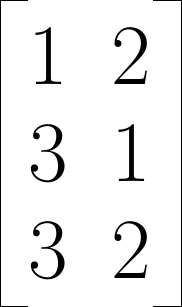}
 		\caption{}
    \end{subfigure}%
    ~ 
    \begin{subfigure}[t]{0.4\textwidth}
        \centering
 		\includegraphics[width=1\columnwidth]{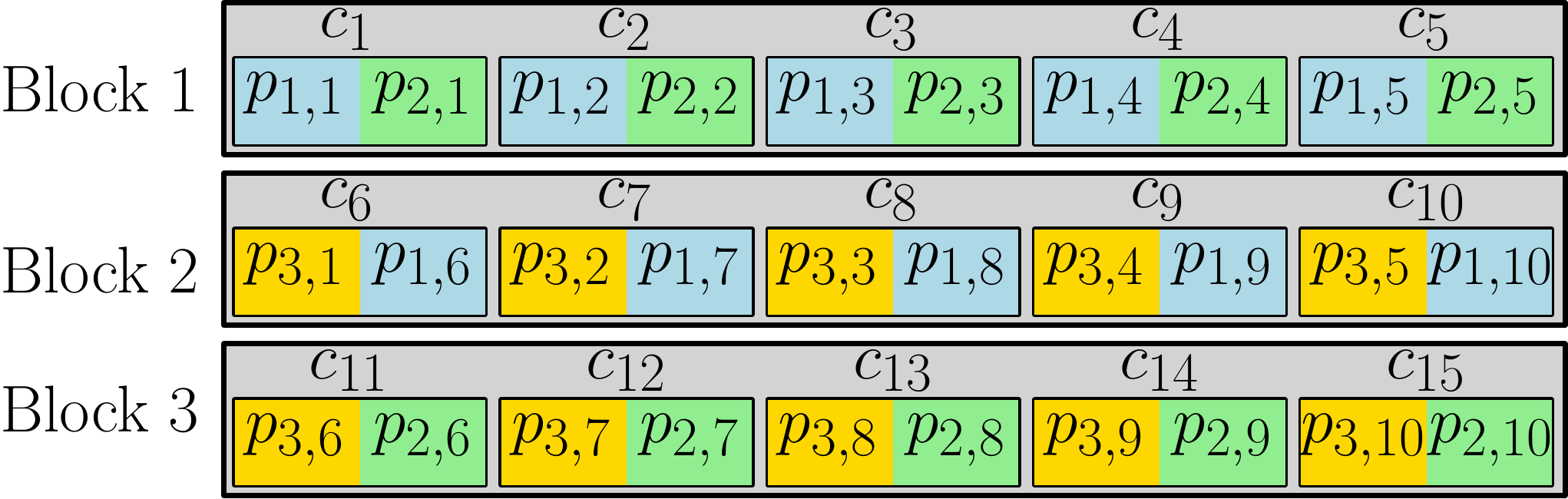}
 		\caption{}
 		
    \end{subfigure}
    \caption{(a) Matrix representation of block design. (b) Three-block BFR-RC.}
    \label{fig:3-block BFR-RC}
    \vspace{-0.2in}
\end{figure*}

 Consider that the file $\fv$ of size $\Mc$ is partitioned into 3 parts $\fv_{1}$, $\fv_{2}$ and $\fv_{3}$ each of size $\tilde{\Mc}=\frac{\Mc}{3}$. Each partition $\fv_i$ is encoded with an $[\tilde{n}=10,\tilde{k}=4,\tilde{d}=5,\tilde{\alpha},\tilde{\beta}]$ regenerating code $\tilde{\Cc}$, where the resulting partition codewords are represented with $\Pc_{1}=\left\{ p_{1,1:\tilde{n}} \right\}$ for $\fv_1$, $\Pc_{2}=\left\{ p_{2,1:\tilde{n}} \right\}$ for $\fv_2$, and $\Pc_{3}=\left\{ p{_{3,1:\tilde{n}}} \right\}$ for $\fv_3$. These symbols are grouped in a specific way and placed into nodes within blocks as represented in Fig.~\ref{fig:3-block BFR-RC}, where each node contains two symbols each coming from two different partitions. We set the BFR code parameters as $[\Mc=3\tilde{\Mc}, k=\frac{3}{2}\tilde{k}, d=2\tilde{d}, \alpha=2\tilde{\alpha}, \beta=\tilde{\beta}]$.


Assume the first block (denoted as Block 1) is unavailable and its first node, which contains codeword $c_{1}$, has to be reconstructed. Due to underlying regenerating code, contacting $5$ nodes of Block 2 and accessing to $p_{1,6:10}$ regenerates $p_{1,1}$. Similarly, $p_{2,1}$ can be reconstructed from Block 3. Any node reconstruction can be handled similarly, by connecting to remaining 2 blocks and repairing each symbol of the failed node by corresponding $\tilde{d}$ nodes in each block. As we have $k=6$, DC, by connecting to 2 nodes from each block, obtains a total of $12$ symbols, which consist of 4 different symbols from each of $\Pc_{1}$, $\Pc_{2}$ and $\Pc_{3}$. As the embedded regenerating code has $\tilde{k}=4$, all $3$ partitions ($\fv_1,\fv_2$ and $\fv_3$) can be recovered, from which $\fv$ can be reconstructed. 


In the following construction, we generalize the BFR-RC construction above utilizing projective planes for the case of having $\rho=0$, $\sigma=1$. As defined in Section \ref{sec:FileSize}, this necessarily requires $d_r > k_c$.) We first introduce projective planes in the following and then detail the code construction.

\begin{definition}[Balanced incomplete block design \cite{Stinson:Combinatorial04}]
A $(v,\kappa,\lambda)$-BIBD has $v$ points distributed into blocks of size $\kappa$ such that any distinct pair of points are contained in $\lambda$ blocks.
\end{definition}

\begin{corollary}\label{thm:BIBDcorollary}
For a $(v,\kappa,\lambda)$-BIBD,
\begin{itemize}
\item Every point occurs in $r=\frac{\lambda (v-1)}{\kappa-1}$ blocks.
\item The design has exactly $b=\frac{vr}{\kappa}=\frac{\lambda(v^2-v)}{\kappa^2-\kappa}$ blocks.
\end{itemize}
\end{corollary}

In the achievable schemes of this work, we utilize a special class of block designs that are called projective planes \cite{Stinson:Combinatorial04}.
\begin{definition}
A $(v=p^2+p+1,\kappa=p+1,\lambda=1)$-BIBD with $p\geq 2$ is called a projective plane of order $p$.
\label{def:projective-plane}
\end{definition}

Projective planes have the property that every pair of blocks intersect at a unique point (as $\lambda=1$). In addition, due to Corollary~\ref{thm:BIBDcorollary}, in projective planes, every point occurs in $r=p+1$ blocks, and there are $b=v=p^2+p+1$ blocks.

\noindent \textbf{\textit{Construction II (Projective plane based placement of regenerating codes):}} For any $(n,b,\Mc,k,$ $\rho=0,\alpha,d,\sigma=1,\beta)$ code satisfying $b=p^2+p+1$, $k=\frac{b}{p+1}\tilde{k}$, $d=(p+1)\tilde{d}$, $\alpha=(p+1)\tilde{\alpha}$, $\beta = \tilde{\beta}$, $p \mid \tilde{d}$, $p+1 \mid \tilde{n}$, $p+1 \mid \tilde{k}$ where $p$ is the order of the underlying projective plane, and $[\tilde{n},\tilde{k},\tilde{d},\tilde{\alpha},\tilde{\beta}]$ represents the underlying regenerating code parameters, consider a file $\fv$ of size $\Mc$.
\begin{itemize}
	\item First divide $\Mc$ into $v=p^2+p+1$ parts, $\fv_{1}$, $\fv_{2}, \cdots, \fv_{v}$.
	\item Each part, of size $\tilde{\Mc}=\frac{\Mc}{v}$, is then encoded using $[\tilde{n},\tilde{k},\tilde{d},\tilde{\alpha},\tilde{\beta}]$ regenerating code $\tilde{\Cc}$. We represent the resulting partition codewords with $\Pc_{i}=p_{i,1:\tilde{n}}$ for $i=1,\dots,v$. We then consider index of each partition as a point in a $(v=p^2+p+1,\kappa=p+1,\lambda=1)$ projective plane. (Indices of symbol sets $\Pc_\Jc$ and points $\Jc$ of the projective plane are used interchangeably in the following.)
	\item We perform the placement of each symbol to the system using this projective plane mapping. (The setup in Fig. \ref{fig:3-block BFR-RC}(b) can be considered as a toy model. Although the combinatorial design with blocks given by $\{p_1,p_2\},\{p_3,p_1\},\{p_3,p_2\}$ has projective plane properties with $p=1$, it is not considered as an instance of a projective plane \cite{Stinson:Combinatorial04}.) In this placement, total of $\tilde{n}$ symbols from each partition $\Pc_{i}$ are distributed to $r$ blocks evenly such that each block contains $\frac{\tilde{n}}{r}$ nodes where each node stores $\alpha=\kappa\tilde{\alpha}$ symbols.
\end{itemize}

Note that blocks of projective plane give the indices of partitions $\Pc_i$ stored in the nodes of the corresponding block in DSS. That is, all nodes in a block stores symbols from unique subset of $\Pc=\{\Pc_1,\dots,\Pc_v\}$ of size $\kappa$. (For instance, in Fig.~\ref{fig:3-block BFR-RC}(b), the first block of the block design has part $\left\{ p_1, p_2\right\}$, and accordingly symbols from partitions $\Pc_1$ and $\Pc_2$ are placed into node of Block 1.) Here, as each point in the block design is repeated in $r$ blocks, the partition codewords span $r$ blocks. Overall, the system can store a file of size $\Mc=v\tilde{\Mc}$ with $b=v$ blocks. (Note that, $r=\kappa=p+1$ and $b=v=p^2+p=1$ for projective planes. See Definition \ref{def:projective-plane}.) We have the parameters as
\begin{equation}
\Mc=v\tilde{\Mc}, k=\frac{b}{r}\tilde{k}, d=\kappa\tilde{d}, \alpha=\kappa\tilde{\alpha}, \beta=\tilde{\beta} 
\label{eq:assignment-b}
\end{equation}
where we choose the parameters to satisfy $r-1 \mid \tilde{d}$, $r \mid \tilde{n}$ (for splitting partition codewords evenly to $r$ blocks) and $r \mid \tilde{k}$ (for data collection as detailed below). We have $d_r=\frac{\tilde{d}}{r-1}=\frac{d}{b-1} > k_c=\frac{\tilde{k}}{r}=\frac{k}{b}$ as $d \geq k$ and hence the required condition $d_r > k_c$ is satisfied

\begin{figure}
 \centering
 \includegraphics[width=1\columnwidth]{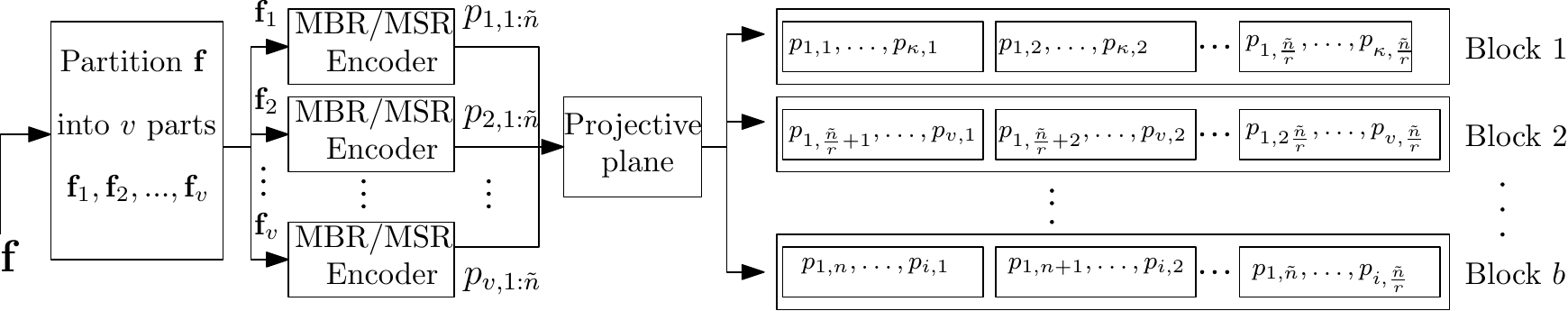}
 \caption{Illustrating the construction of BFR codes using projective plane based placement of regenerating codes. ($n'=\tilde{n}-\frac{\tilde{n}}{r}+1.)$}
\label{fig:general_projective}
\vspace{-0.3in}
\end{figure}

\emph{Node Repair:} Consider that one of the nodes in a block is to be repaired. Note that the failed node contains $\kappa$ symbols, each coming from a distinct partition. Using the property of projective planes that any $2$ blocks has only $1$ point in common, any remaining block can help for the regeneration of $1$ symbol of the failed node. Furthermore, as any point in the block design has a repetition degree of $r$, one can connect to $r-1$ blocks, $d_r=\frac{\tilde{d}}{r-1}$ nodes per block, to repair one symbol of a failed node. Combining these two observations; we observe that node regeneration can be performed by connecting to $(r-1)\kappa$ blocks. In particular, substituting $r=\kappa=p+1$, we see that connecting to $p^{2}+p=b-1$ blocks allows for reconstructing of any node of a failed block.


\emph{Data Collection:} DC connects to $k_c=\frac{\tilde{k}}{r}$ nodes per block from all $b_{c}=b$ blocks, i.e., a total of $k=\frac{b}{r}\tilde{k}=\frac{v}{r}\tilde{k}$ nodes each having encoded symbols of $\kappa=r$ partitions. These total of $v\tilde{k}$ symbols include $\tilde{k}$ symbols from each partition, from which all partitions can be decoded, and hence the file $\fv$, can be reconstructed.


\begin{remark}
 This construction is related to \textit{layered codes} studied in \cite{Tian:Layered14}. In that work, layering helps to construct codes with exact repair properties. In Construction II, on the other hand, multiple nodes in the system can have the same \textit{type} (representing the same subset of partitions), and this enables to achieve different operating points for the block failure model.
\end{remark}

\subsubsection{BFR-MSR}
To construct a BFR-MSR code, we set each sub-code $\tilde{\Cc}$ in Construction II as an MSR code, which has $ 
\tilde{\alpha}=\frac{\tilde{\Mc}}{\tilde{k}}$ and $\tilde{d}\tilde{\beta}=\frac{\tilde{\Mc}\tilde{d}}{\tilde{k}(\tilde{d}-\tilde{k}+1)}$. This, together with \eqref{eq:assignment-b}, results in the following parameters of our BFR-MSR construction
\begin{equation}
\alpha=\tilde{\alpha}\kappa=\frac{\Mc}{k}, d\beta=\kappa \tilde{d}\tilde{\beta}=\frac{\Mc d}{k(d-\frac{k(p+1)^{2}}{p^2+p+1}+p+1)}.
\label{eq:bfr-msr}
\end{equation}

We remark that if we utilize ZigZag codes\cite{Tamo:Zigzag13} as the sub-code $\tilde{\Cc}$ above, we have $[\tilde{n},\tilde{k},\tilde{d}=\tilde{n}-1,\tilde{\alpha}=\tilde{r}^{\tilde{k}-1},\tilde{\beta}=\tilde{r}^{\tilde{k}-2}, \tilde{r}=\tilde{n}-\tilde{k}]$, and having $\tilde{d}=\tilde{n}-1$ requires connecting to $1$ node per block for repairs in our block model. In addition, product matrix MSR codes \cite{Rashmi:Optimal11} require $\tilde{d} \geq 2\tilde{k}-2$, and they can be used as the sub-code $\tilde{\Cc}$, for which we do not necessarily have $\tilde{d}={r-1}$. We observe from \eqref{eq:min-cut_b-MSR-values} and \eqref{eq:bfr-msr} that the BFR-MSR point is achieved for $\tilde{k}=p+1$, implying $k=b$, i.e. DC connects necessarily 1 node per block for data reconstruction when our Construction II gives BFR-MSR code. 
 

\subsubsection{BFR-MBR}

To construct a BFR-MBR code, we set each sub-code $\tilde{\Cc}$ in Construction II as a product matrix MBR code \cite{Rashmi:Optimal11}, which has $\tilde{\alpha}=\tilde{d}\tilde{\beta}= \frac{2\tilde{\Mc}\tilde{d}}{\tilde{k}(2\tilde{d}-\tilde{k}+1)}$.
This, together with \eqref{eq:assignment-b}, results in the following parameters of our BFR-MBR construction
\begin{equation}
\alpha=d\beta=\frac{2\Mc d}{k(2d-\frac{k(p+1)^2}{p^2+p+1}+p+1)}.
\label{eq:bfr-mbr}
\end{equation}

From \eqref{eq:min-cut_b-MBR-values} and \eqref{eq:bfr-mbr}, we observe that the BFR-MBR point is achieved for $\tilde{k}=p+1$.

\subsection{Duplicated Block Design Based BFR Codes ($\rho=0$ and $\sigma < b-1$)}
\label{subsec:dcbd}
In this section, BFR codes for special case of having $\rho=0$ is constructed. We note that $\rho=0$ implies that DC contact all $b$ blocks to retrieve the stored data. Before detailing the code construction, we first introduce a block design referred to as duplicated combination block design \cite{Tian:Layered14}.


\begin{definition}[DCBD]
Let $(\tilde{\kappa},\tilde{v})$ denote the parameters for a block design, where $\tilde{v}$ points from all possible sets of blocks each with $\tilde{\kappa}$ points. Then, duplicated combination block design (DCBD) (with repetition $\tilde{r}$ ) is a block design where the given block design is duplicated $\tilde{r}$ times with different labeling points. (Here, total of $v=\tilde{r}\tilde{v}$ points are splitted into $\tilde{r}$ groups of $\tilde{v}$ points, where each group generates sub-blocks according to the given block design.)

\end{definition}

\begin{example}
DCBD with $\tilde{v}=5$, $\tilde{\kappa}=4$ and $\tilde{r}=3$ is given below.

\begin{equation}
\left[ \begin{array}{@{}*{12}{c}@{}}
     1 & 3 & 4 & 5 & \: \quad 6 & 8 & 9 & 10 & \: \quad 11 & 13 & 14 & 15\\
     1 & 2 & 4 & 5 & \quad 6 & 7 & 9 & 10 & \: \quad 11 & 12 & 14 & 15\\
     1 & 2 & 3 & 5 & \quad 6 & 7 & 8 & 10 & \: \quad 11 & 12 & 13 & 15\\
     1 & 2 & 3 & 4 & \quad 6 & 7 & 8 & 9 & \: \quad 11 & 12 & 13 & 14\\
     2 & 3 & 4 & 5 & \quad 7 & 8 & 9 & 10 & \: \quad 12 & 13 & 14 & 15\\
\end{array} \right]
\end{equation}


It can be observed that each sub-block consists $\tilde{v} \choose \tilde{\kappa}$ blocks, each containing a different set of $\tilde{\kappa}$ points. Also, the same combination is repeated $\tilde{r}$ times (with different labels for points, namely $\left\{ 6,7,8,9,10\right\}$ and $\left\{ 11,12,13,14,15\right\}$). Each row here corresponds to a block of DCBD, where sub-blocks aligned similarly in columns represent the underlying $(\tilde{\kappa},\tilde{v})$ block design. We refer to the sub-blocks as \textit{repetition groups} in the following.
\end{example}

\noindent \textbf{\textit{Construction III (DCBD based BFR-RC):}} For any $(n,b,\Mc,k,$ $\rho=0,\alpha,d,\sigma<b-1,\beta)$ code satisfying  $b={\tilde{v}\choose\tilde{\kappa}}$, $k=\frac{b}{b-1}\tilde{k}$, $d=\frac{(b-\sigma)\tilde{d}}{b-\sigma-1}$, $\alpha=(b-1)(b-\sigma)\tilde{\alpha}$, $\beta = (b-\sigma-1)(b-1)\tilde{\beta}$, $b-\sigma-1 \mid \tilde{d}$, $b-1 \mid \tilde{n}$, $b-1 \mid \tilde{k}$ where $(\tilde{v}, \tilde{\kappa})$ are the underlying DCBD parameters and  $[\tilde{n},\tilde{k},\tilde{d},\tilde{\alpha},\tilde{\beta}]$ represents the underlying regenerating code parameters, consider a file $\fv$ of size $\Mc$.

\begin{itemize}
\item Divide $\Mc$ into $(b-\sigma){b \choose b-1}$ parts of equal size $\tilde{\Mc}$, i.e., $\tilde{\Mc}(b-\sigma)b=\Mc$. 
\item \looseness=-1 Encode each part $\fv_i$ using an $[\tilde{n},\tilde{k}=\tilde{\Mc},\tilde{d}]$ regenerating code (referred to as the sub-code $\tilde{\Cc}$). 
\item Place the resulting partition codewords according to DCBD design (with $\tilde{v}=b$, $\tilde{\kappa}=b-1$ and $\tilde{r}=b-\sigma$) such that each block has $c=\frac{\tilde{n}}{b-1}$ nodes, where each node stores $\kappa=\tilde{\kappa}\tilde{r}$ symbols, each coming from a different partition.
\end{itemize}

The system stores a file of size $\Mc=b(b-\sigma)\tilde{\Mc}$ over $b$ blocks. We have the parameters as 
\begin{equation}
\Mc=b(b-\sigma)\tilde{\Mc}, k=\frac{b\tilde{k}}{b-1}, d=\frac{\tilde{d}(b-\sigma)}{b-\sigma-1}, \alpha=(b-1)(b-\sigma)\tilde{\alpha}, \beta=(b-\sigma-1)(b-1)\tilde{\beta}
\label{eq:assignment-dcbd}
\end{equation}
where we consider $\sigma < b-1$.

The following example (with $b=5$ and $\sigma=2$) illustrates a repair scenario. Assume that the failed node is in the first block and it will be regenerated by blocks 2,3 and 4 (as $b-\sigma=3$). Considering the first repetition group, it can be observed that symbols of each of $b-\sigma=3$ partitions ($\Pc_3, \Pc_4$ and $\Pc_5$) can be found in $b-\sigma-1=2$ of these blocks, whereas the remaining $\sigma-1=1$ partition ($\Pc_1$) can be found in all $b-\sigma$ blocks. (In the representation below, numbers represent the indices for partitions $\Pc_i$. And, the three highlighted rows for the first repetition group includes partitions $\Pc_1,\Pc_3,\Pc_4,\Pc_5$ that are relevant to the symbols stored in the block corresponding to the first row below.)


\begin{equation}
\left[ \begin{array}{@{}*{12}{c}@{}}
     1 & 3 & 4 & 5 & \: \quad 6 & 8 & 9 & 10 & \: \quad 11 & 13 & 14 & 15\\
     \htext{1} & 2 & \htext{4} & \htext{5} & \: \quad 6 & 7 & \htext{9} & \htext{10} & \: \quad \htext{11} & 12 & \htext{14} & \htext{15} \\
     \htext{1} & 2 & \htext{3} & \htext{5} & \: \quad \htext{6} & 7 & \htext{8} & \htext{10} & \: \quad 11 & 12 & \htext{13} & \htext{15} \\
     1 & 2 & \htext{3} & \htext{4} & \: \quad \htext{6} & 7 & \htext{8} & \htext{9} & \: \quad \htext{11} & 12 & \htext{13} & \htext{14} \\
     2 & 3 & 4 & 5  & \: \quad 7 & 8 & 9 & 10 & \: \quad 12 & 13 & 14 & 15\\
\end{array} \right]
\label{eq:DCBD-ex-reg}
\end{equation}

\textit{Node Repair:} Generalizing above argument, consider that one of the nodes in a block is to be repaired by contacting to $b-\sigma$ blocks. A failed node contains $\kappa=(b-1)(b-\sigma)$ symbols, each coming from a distinct partition codeword. The properties of the underlying (DCBD) block design, (considering the first repetition group), implies that there exists $b-\sigma$ partitions of the failed node that are contained only in $b-\sigma-1$ of the blocks contacted for repair. The remaining $b-1-(b-\sigma)=\sigma-1$ partitions (of each repetition group) are contained in all of the contacted $b-\sigma$ blocks. These partitions are referred to as the \textit{common partitions} of a repetition group in $b-\sigma$ contacted blocks. (In the example above, partitions $\Pc_1, \Pc_6, \Pc_{11}$ are the common partitions for the first, second and third repetition group respectively.)


Using this observation for DCBD based construction, (i.e., considering all repetition groups), consider obtaining $\sigma-1$ common partitions from $b-\sigma-1$ blocks of each of the $\tilde{r}=b-\sigma$ repetition groups. In addition, consider obtaining remaining relevant partitions ($b-\sigma$ partitions per repetition group) from these $\tilde{r}=b-\sigma$ partition groups (total of $\tilde{r}(b-\sigma)(b-\sigma-1)$ partitions). 

These $\sigma-1$ common partitions per repetition group over $\tilde{r}=b-\sigma$ repetition groups are contacted evenly. Namely, each other relevant partition are contacted from only $b-\sigma-1$ blocks, the common partitions among $b-\sigma$ blocks are contacted only $b-\sigma-1$ times (i.e., by not contacting any common point at all in only one repetition group from a block and since there are $b-\sigma$ repetition groups and $b-\sigma$ contacted blocks, we can do this process evenly for all blocks). Henceforth, from each block same amount of symbols (and same amount of symbols from each partition) is downloaded. In total, there are $(\sigma-1)(b-\sigma)(b-\sigma-1)$ common points and each block contributes the transmission of $(\sigma-1)(b-\sigma-1)$ common partitions. Hence, $\beta=[(\sigma-1)(b-\sigma-1)+ (b-\sigma-1)(b-\sigma)]\tilde{\beta}=(b-\sigma-1)(b-1)\tilde{\beta}$.


In order to have a successful regeneration for each node to be regenerated, we also require the following condition in this construction.

\begin{lemma}
Construction III requires the necessary condition $\frac{\tilde{n}}{b-1} \geq \frac{\tilde{d}}{b-\sigma-1}$ for repair feasibility.
\label{lmm:dcbd_rep}
\end{lemma}
\begin{proof}
Given $v-1$ combinations of $v$ points, any two combinations differs only in one point. Also, any combination is missing only one point. If one collects $b-\sigma\geq 2$ of such combinations in Construction III, then the partition with least number of instances is $b-\sigma-1$. (This follows as the first combination is missing only one partition which is necessarily included in the second combination.) Then, by contacting any $b-\sigma$ blocks, one can recover partitions of failed node from these $b-\sigma$ blocks. (There exist at least $b-\sigma-1$ number of blocks containing nodes storing symbols from a given partition.) Since each block has $c=\frac{\tilde{n}}{b-1}$ symbols from each partition, we require $\frac{\tilde{n}}{b-1} \geq \frac{\tilde{d}}{b-\sigma-1}$ to have repair feasibility in Construction III.
\end{proof}

Therefore, a failed node can be regenerated from $d=\frac{\tilde{d}(b-\sigma)}{b-\sigma-1}$ nodes and downloading $\beta=(b-\sigma-1)(b-1)\tilde{\beta}$ symbols from each block. (Note that, $d_r=\frac{d}{b-\sigma}=\frac{\tilde{d}}{b-\sigma-1}$.) For each partition of the failed node, $\tilde{d}\tilde{\beta}$ symbols are downloaded, from which one can regenerate each partition. Note that, repeating combinations multiple times enables us to have uniform downloads from the nodes during repairs.


\textit{Data Collection:} DC connects to $k_c=\frac{k}{b}=\frac{\tilde{k}}{b-1}$ nodes per block (as $\rho=0)$, and downloads total of $k_c\alpha$ symbols from each block. These symbols include $\frac{\tilde{k}\tilde{\alpha}}{b-1}$ symbols from each of $(b-1)(b-\sigma)$ partitions. Therefore, from all blocks, $\frac{(b-1)\tilde{k}\tilde{\alpha}}{b-1}=\tilde{k}\tilde{\alpha}$ symbols per partition is collected, from which each partition can be decoded via underlying sub-code $\tilde{\Cc}$, and the stored file $\fv$ can be reconstructed.   

\subsubsection{BFR-MSR}
To construct a BFR-MSR code, we set each sub-code $\tilde{\Cc}$ in Construction III as an MSR code, which has $\tilde{\alpha}=\frac{\tilde{\Mc}}{\tilde{k}}$ and  $\tilde{d}\tilde{\beta}=\frac{\tilde{\Mc}\tilde{d}}{\tilde{k}(\tilde{d}-\tilde{k}+1)}$. This, together with \eqref{eq:assignment-dcbd}, results in the following parameters of our BFR-MSR construction
\begin{eqnarray}
\alpha &=&(b-1)(b-\sigma)\tilde{\alpha}=\frac{\Mc}{k},\label{eq:bfr-msr-DCBD-alp}\\
d\beta &=&\frac{\tilde{d}(b-\sigma)}{b-\sigma-1}(b-\sigma-1)(b-1)\tilde{\beta} =\frac{\Mc d (b-\sigma-1)}{k(d(b-\sigma-1)-\frac{k(b-1)(b-\sigma)}{b}+b-\sigma)}.
\label{eq:bfr-msr-DCBD-bet}
\end{eqnarray}

From \eqref{eq:case2_MSR}, we obtain that Construction III results in optimal BFR-MSR codes when $k=\frac{b}{\sigma}$ (i.e.,  $\tilde{k}=\frac{b-1}{\sigma}$).

\subsubsection{BFR-MBR}

To construct a BFR-MBR code, we set each sub-code $\tilde{\Cc}$  in Construction III as a product matrix MBR code \cite{Rashmi:Optimal11}, which has $\tilde{\alpha}=\tilde{d}\tilde{\beta}= \frac{2\tilde{\Mc}\tilde{d}}{\tilde{k}(2\tilde{d}-\tilde{k}+1)}$.
This, together with \eqref{eq:assignment-dcbd} results in the following parameters of our BFR-MBR construction
\begin{equation}
\alpha=d\beta=\frac{2 \Mc d (b-\sigma-1)}{k(2d(b-\sigma-1)-\frac{k(b-1)(b-\sigma)}{b}+b-\sigma)}.
\label{eq:bfr-mbr-DCBD}
\end{equation}

From \eqref{eq:case2_MBR}, we obtain that Construction III results in optimal BFR-MSR codes when $k=\frac{b^2}{b+\sigma^2-1}$ (i.e., $\tilde{k}=\frac{b(b-1)}{b+\sigma^2-1}$).

\section{Locally Repairable BFR Codes}
\label{sec:BFR-LRC}

\subsection{Locality in BFR}

In this section, we focus on BFR model with repair locality constraints, i.e., only a local set of blocks are available to regenerate the content of a given node. This model is suitable for both disk storage and distributed (cloud/P2P) storage systems. For example, a typical data center architecture includes multiple servers, which form racks which further form clusters, see Fig.~\ref{fig:local}. We can think of each cluster as a local group of racks where each rack contains multiple servers. Hence, in Fig.~\ref{fig:local}, we can model the data center as having $b$ blocks (racks) where $b_L$ blocks form a local group (cluster) and there are $c$ nodes (servers) in each block. 

In this section, we extend our study of data recovery with the block failure model to such scenarios with locality constraints. We assume that DSS maintains the data of size $\Mc$ with at most $\rho$ blocks being unavailable. Hence, from any $b-\rho$ blocks, data collection can be performed by contacting any $k_c$ nodes from each of such blocks. In other words, DC can contact some set of local groups, $\Bc^*$,  to retrieve the file stored in the DSS with $|\Bc^*|=b-\rho$. Let $\cv_{i,j}$ denote the part of the codewords associated with $i^{\rm{th}}$ local group's $j^{\rm{th}}$ block accessed by DC, which consists of $k_c$ nodes. Note that, we consider $k_c=c$ for full access and $k_c< c$ for partial access as before. Then, we can denote by $\cv_i$ the codeword seen by DC corresponding to local group $\Bc_i$, which has a size $j' \leq b_L$. Therefore, $\cv_{\Bc^*}=\left\{(\cv_{i_1},\cdots,\cv_{i_{|\Bc^*|}})\right\}$ denotes the codeword corresponding to the one seen by DC, when accessing $k_c(b-\rho)=k$ nodes.


\begin{figure}
 \centering
 \includegraphics[width=0.5\columnwidth]{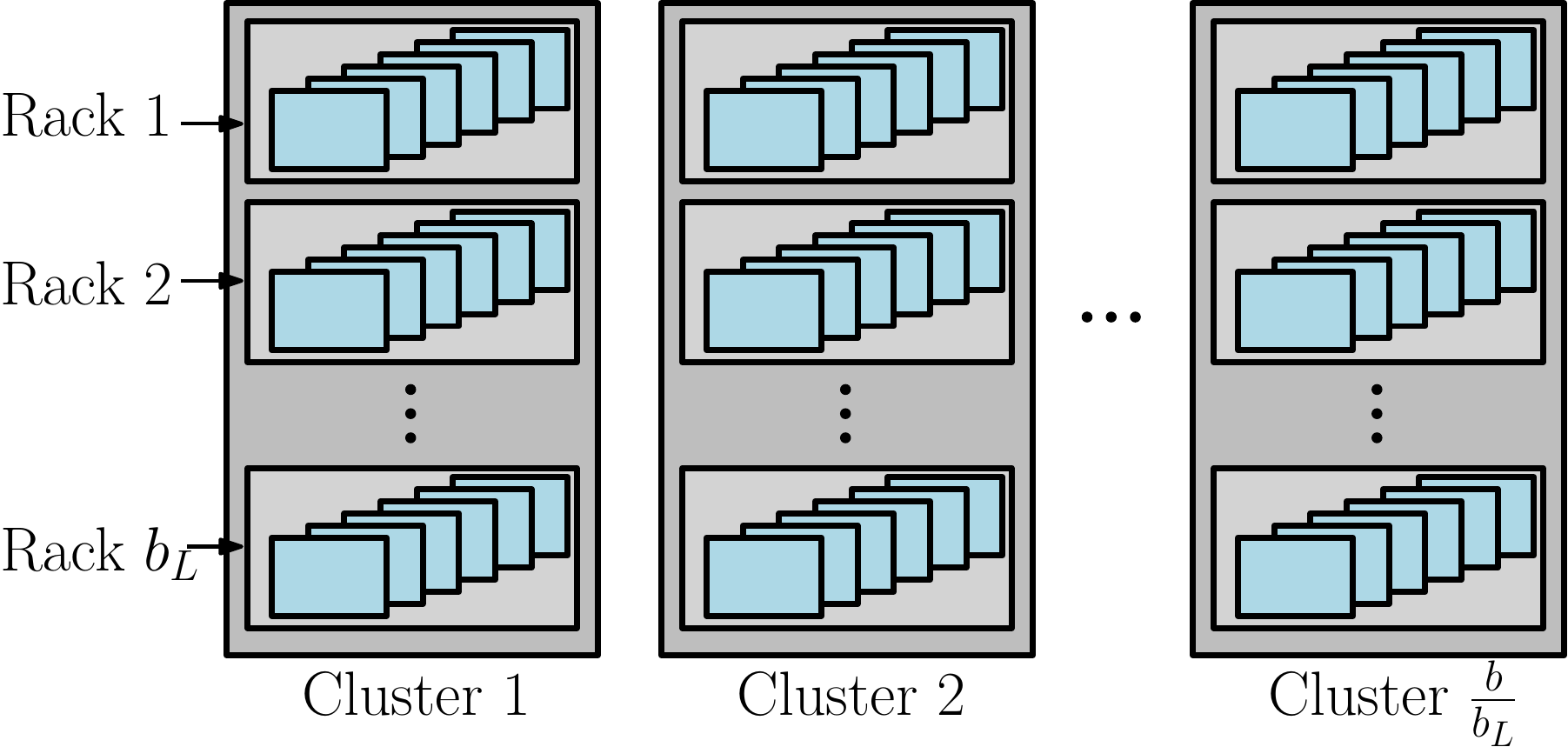}
 \caption{Data center architecture. Each rack resembles a block and each cluster forms a local group for node repair operations in the BFR model.}
\label{fig:local}
\vspace{-0.3in}
\end{figure}


\begin{definition} 
Let $\cv$ be a codeword in $\Cc$ selected uniformly, the resilience of $\Cc$ is defined as
\begin{equation}
\rho=b-\max\limits_{\Bc^* \subseteq[b]:H(\cv(\Bc^*)) < \Mc}|\Bc^*| - 1.
\label{def:resilience}
\end{equation}
\end{definition}


We remark that the resilience $\rho$ of a locally repairable BFR code dictates the number of block failures that the system can tolerate (analogous to the minimum distance providing the maximum number of node failures). In addition, DC can do data access by contacting any $b-\rho$ blocks. 

\begin{definition}
Code $\Cc$ is said to have $(r_L,\rho_L)$ locality, if for any node $j$ in a block $i$, there exists a set of blocks $\Bc_i$ such that
\begin{itemize}
\item $i \in \Bc_i\subset \Bc=\{1,\cdots,b\}$,
\item $|\Bc_i|\leq b_L := r_L + \rho_L$,
\item $\Cc|_{\Bc_i}$ is an $(n_L,b_L,K_L,k_L,\rho_L,\alpha)$ BFR code. (See Definition \ref{def:BFR}.)
\end{itemize}  
\end{definition}

Note that $K_L$ is the size of the local data for corresponding BFR local code (i.e., $\Mc$ in Definition~\ref{def:BFR}). Such codes are denoted as BFR-LRC in the following. 


\subsubsection{Upper bound on the resilience of BFR-LRC}

We provide the following bound on the resilience of BFR-LRC.

\begin{theorem}
The resilience of BFR-LRC can be bounded as follows.
\begin{equation}
\rho \leq b - \ceilb{\frac{\Mc(b_L-\rho_L)}{K_L}} - \left(\ceilb{\frac{\Mc}{K_L}}-1\right)\rho_L
\end{equation}
\label{thm:BFR-LRC}
\end{theorem}
\begin{proof}
The proof follows from the algorithmic approach as considered in \cite{Gopalan:Locality12}, \cite{Papailiopoulos:Locally12} and \cite{Rawat:Optimal14}, and detailed in the Appendix \ref{app:res_theo}. 
\end{proof} 

\subsection{Code Constructions for Resilience Optimal BFR-LRC}

In this section we propose two code constructions, which yield optimal codes in terms of resilience. Both of the constructions utilize Gabidulin coding as well as MDS codes but only the first construction uses projective plane geometry. 

\subsubsection{Basic Construction with Projective Planes}
~\\
\noindent \textbf{\textit{Construction IV:}} Consider a file $\fv$ of size $\Mc$ and projective plane of order $p$, $(v,\kappa,\lambda=1)-BIBD$.
\begin{itemize}
\item First, encode $\fv$ using $[N=\frac{K_Lb}{b_L},K=\Mc,D=N-\Mc+1]$ Gabidulin code, $\Cc^{\rm{Gab}}$. The resulting codeword $\cv \in \Cc^{\rm{Gab}}$ is divided into $\frac{b}{b_L}$ disjoint sets of symbols of size $K_L$.
\item For each disjoint set, divide it into $v$ partitions of equal size. For each partition, use $[\tilde{n},\tilde{k}=\frac{k_L(p+1-\rho_L)}{b_L-\rho_L}]$ MDS code, where $n_L=\tilde{b}v$.
\item Place the resulting encoded symbols using \emph{Construction II} (with projective plane of order p). That is, local code in the local group will have BFR properties and constructed as in \emph{Construction II}. 
\end{itemize}

\begin{remark}
\textit{Construction IV} works only if $\rho_L \leq p$, since otherwise we would have non-positive value for $\tilde{k}$. Furthermore, we need $(b_L-\rho_L) \mid k_L(p+1-\rho_L)$ to have an integer $\tilde{k}$.
\end{remark}

\begin{corollary}
\textit{Construction IV} provides resilience-optimal codes $\Cc^{\rm{BFR-LRC}}$, when $p^2+p+1 \mid K_L \mid \Mc$ and $b_L \mid b$.
\label{cor:cons1}
\end{corollary}

\begin{proof}
In order to prove that $\Cc^{\rm{BFR-LRC}}$ attains the bound in \eqref{thm:BFR-LRC}, it is sufficient to demonstrate that any pattern of $E = b - \ceilb{\frac{\Mc(b_L-\rho_L)}{K_L}} - \left(\ceilb{\frac{\Mc}{K_L}}-1\right)\rho_L$ number of block erasures can be corrected. For $b=\frac{Nb_L}{K_L}$, we then have $E=\frac{Nb_L}{K_L}-\ceilb{\frac{\Mc(b_L-\rho_L)}{K_L}} - \left(\ceilb{\frac{\Mc}{K_L}}-1\right)\rho_L$. Also note that the worst case erasure pattern is the one when erasures happen in the smallest possible number of local groups and the number of block erasures inside a local group is the maximum possible.

Let $\Mc=\alpha_1K_L$, then $E=b_L(\frac{N}{K_L}-\alpha_1)+\rho_L$, which corresponds to $\frac{N}{K_L}-\alpha_1$ local groups with $b_L$ erasures and $\rho_L$ erasures from one additional group (which results in no rank erasures for the underlying Gabidulin codewords since resilience of local group is $\rho_L$). Then, total rank erasures is $K_L(\frac{N}{K_L}-\alpha_1)=N-K_L\alpha_1$ which can be corrected by Gabidulin code since its minimum distance is $D=N-\Mc+1=N-\alpha_1K_L+1$. 
\end{proof}

\begin{remark}
Although above construction yields optimal codes in terms of resilience, when $\rho_L \not = 0$, they are not rate optimal. The optimal codes should store $k_L \alpha = k_Lv \tilde{\alpha}=k_L (p+1) \tilde{\alpha}$ symbols (in each local group), since each node stores $\alpha=(p+1)\tilde{\alpha}$ where $\tilde{\alpha}$ is the number of symbols from a partition and each node stores symbols from $p+1$ partitions. On the other hand, with the above construction, one can store $\tilde{k} v \tilde{\alpha} = \frac{k_L(p+1-\rho_L)}{p^2+p+1-\rho_L}(p^2+p+1)\tilde{\alpha}$ symbols (in each local group). Note that for $\rho_L \not = 0$, $\frac{k_L(p+1-\rho_L)}{p^2+p+1-\rho_L}(p^2+p+1)\tilde{\alpha} \leq k_L (p+1) \tilde{\alpha}$, which implies that the code construction is not optimal in terms rate (i.e., file size for a fixed $n$).
\end{remark}

\subsubsection{Construction for BFR-LRC with Improved Rate}

In the previous section, we utilize projective planes in code construction, however resulting codes are not optimal in terms of the rate. We now propose another construction for BFR-LRC to achieve higher rate.

\noindent \textbf{\textit{Construction V:}} Consider a file $\fv$ of size $\Mc$. 
\begin{itemize}
\item First, encode $\fv$ using $[N,K=\Mc,D=N-\Mc+1]$ Gabidulin code, $\Cc^{\rm{Gab}}$. The resulting codeword $\cv \in \Cc^{\rm{Gab}}$ is divided into $\frac{b}{b_L}$ disjoint groups of size $K_L = \frac{Nb_L}{b}$.
\item Apply MDS codes of $[\tilde{n}=b_Lc,\tilde{k}=\frac{Nb_L}{b}]$ to each disjoint group. Resulting $b_Lc$ symbols are placed to each block equally ($c$ symbols per block).
\end{itemize}

\begin{corollary}
\textit{Construction V} yields an optimal code, $\Cc^{\rm{BFR-LRC}}$, when $b \mid Nb_L$.
\label{cor:cons2}
\end{corollary}

\begin{proof}
Let $\alpha_1$, $\beta_1$ and $\gamma_1$ be integers such that $\Mc = \frac{K_L}{b_L-\rho_L}(\alpha_1(b_L-\rho_L)+\beta_1)+\gamma_1$, where $1 \leq \alpha_1 \leq \frac{b}{b_L}$; $0 \leq \beta_1 \leq b_L-\rho_L-1$; and $0 \leq \gamma_1 \leq \frac{K_L}{b_L-\rho_L}-1$. $E = b - \ceilb{\frac{\Mc(b_L-\rho_L)}{K_L}} - \left(\ceilb{\frac{\Mc}{K_L}}-1\right)\rho_L$ is the number of block erasures that can be tolerated.

\begin{itemize}
\item If $\gamma_1=\beta_1=0$, then $\Mc=K_L\alpha_1$. Then, we have $E=b_L(\frac{N}{K_L}-\alpha_1)+\rho_L$ number of block erasures to be tolerated, similar to Corollary~\ref{cor:cons1}. Thus, the worst case happens when $\frac{N}{K_L}-\alpha_1$ local groups with all of their blocks erased  and one additional local group with $\rho_L$ blocks erased. The latter does not correspond to any rank erasures since the resilience of local group is $\rho_L$, therefore the worst case scenario results in $K_L$ rank erasures in each of $(\frac{N}{K_L}-\alpha_1)$ local groups. Since $D-1=N-\Mc=N-K_L\alpha_1$, these worst case scenarios can be corrected by the Gabidulin code.

\item If $\gamma_1=0$ and $\beta_1 > 0$, then $\Mc = \frac{K_L}{b_L-\rho_L}(\alpha_1(b_L-\rho_L)+\beta_1)$. Hence, we have $E=b_L(\frac{N}{K_L}-\alpha_1-1)+b_L-\beta_1$. Thus, the worst case happens when $\frac{N}{K_L}-\alpha_1-1$ local groups with all of their blocks erased and one additional local group with $b_L-\beta_1$ blocks erased. Then, that scenario corresponds to $(\frac{N}{K_L}-\alpha_1)K_L-\beta_1\frac{K_L}{b_L-\rho_L}$ rank erasures, which can be corrected by Gabidulin code since $D-1=N-\Mc=N-\alpha_1K_L-\beta_1\frac{K_L}{b_L-\rho_L}$.

\item If $\gamma_1 > 0$, then $\Mc = \frac{K_L}{b_L-\rho_L}(\alpha_1(b_L-\rho_L)+\beta_1)+\gamma_1$. Hence, we have $E=b_L(\frac{N}{K_L}-\alpha_1-1)+b_L-\beta_1-1$. Thus, the worst case happens when $\frac{N}{K_L}-\alpha_1-1$ local groups with all of their blocks erased and one additional local group with $b_L-\beta_1-1$ blocks erased. Then, that scenario corresponds to $(\frac{N}{K_L}-\alpha_1)K_L-(\beta_1+1)\frac{K_L}{b_L-\rho_L}$ rank erasures, which can be corrected by Gabidulin code since $D-1=N-\Mc=N-\alpha_1K_L-\beta_1\frac{K_L}{b_L-\rho_L}-\gamma_1$ and $\gamma_1 < \frac{K_L}{b_L-\rho_L}$.
\end{itemize}
\end{proof}

\begin{remark}
We note that the above construction is similar to that of \cite{Rawat:Optimal14}, here modified for the block failure model to achieve resilience optimal construction with the improved rate. 
\end{remark}

\subsection{Local Regeneration for BFR Codes}

In the previous section, MDS codes are utilized in local groups. These codewords, however, are not repair bandwidth optimal. As an alternative, regenerating codes can be used in local groups to construct codes which have better trade-off in terms of repair bandwidth. Differentiating between the two important points, we denote BFR-LRC codes with regenerating code properties which operate at minimum per-node storage point as BFR-MSR-LRC. Similarly, the codes operating at minimum repair bandwidth point are called BFR-MBR-LRC.

Let $\Gc_1,\dots,\Gc_{\frac{b}{b_L}}$ represent the disjoint set of indices of blocks where $\Gc_i$ represents local group $i$, which has $b_L$ blocks. A failed node in one of the blocks in $\Gc_i$ is repaired by contacting any $b_L-\sigma_L$ blocks within the group. A newcomer downloads $\beta$  symbols from $\frac{d_L}{b_L-\sigma_L}$ nodes from each of $b_L-\sigma_L$ blocks. That is, the local group has the properties of repair bandwidth efficient BFR codes as studied in Section~\ref{sec:FileSize}.


\begin{definition}[Uniform rank accumulation (URA) codes]
Let $G$ be a generator matrix for a BFR code $\Cc$. Columns of $G$ produce the codewords, henceforth we can think of each $\alpha$ columns (also referred to as \emph{thick} column) of $G$ as a representation of a node storing $\alpha$ symbols. Then, a block can be represented by $c$ such \emph{thick} columns. Let $S_i$ be an arbitrary subset of $i$ such blocks. $\Cc$ is an URA code, if the restriction $G|_{S_i}$ of $G$ to $S_i$, has rank $\rho_i$ that is independent of specific subset $S_i$ and it only depends on $|S_i|$.
\end{definition}

\begin{remark}
\looseness=-1 Generally, URA codes are associated with rank accumulation profile to calculate the rank of any subset $S_i$. However, rank accumulation profile is not required but rather rank accumulation is enough for a code to be considered as URA code. We note that for BFR-MSR/MBR codes, we do not have a specific rank accumulation profile but that does not rule out BFR-MSR/MBR codes being URA since they still obey the rank accumulation property. Specifically, we note that $H(b_{\Sc_i})=f(|\Sc_i|)$ for both BFR-MSR/MBR, which makes them URA codes.  
\end{remark}

Following similar steps as introduced in \cite{Kamath:Explicit13}, resilience upper bound can be derived when local codes are URA codes. Consider the finite length vector $(b_1,\dots,b_{b_L})$ and its extension with $b_L$ period as $b_{i+jb_L}=b_i$, $1\leq i \leq b_L$, $j\geq 1$. Let $H(b_\Sc)$ denote the entropy of set of blocks $\Sc$,
\begin{equation}
H(b_\Sc)=\sum_{i=1}^{|\Sc|}H(b_{j_i}|b_{j_1},\cdots,b_{j_{i-1}}), \:\:\:\:\:\: |\Sc| \geq 1, \quad \Sc = \left\{ j_1,\dots,j_{|S|}\right\} \subseteq [b].
\end{equation}
$H(\cdot)$ function in the above only depends on the structure of $\Sc$, which contains $\mu$ number of local groups each with $b_L$ blocks and additional set of $\phi$ blocks. We denote the entropy of this set by referring to this structure and use $H(\cdot)$ (with some abuse of notation) function with an argument of size of $\Sc$. More specifically, for integers $\mu \geq 0$ and $1 \leq \phi \leq b_L$, let $H(\mu b_L+\phi)=\mu K_L+H(\phi)$. (Note that, due to URA property, entropy is only a function of number of blocks here.) For the inverse function $H^{(\textrm{inv})}$, we set $H^{(\textrm{inv})}(\varphi)$ for $\varphi \geq 1$, to be largest integer $\Sc$ such that $H(b_\Sc)\geq \varphi$. Then, we have for $\tilde{\mu}\geq 0$ and $1 \leq \tilde{\phi} \leq K_L$, $H^{(\textrm{inv})}(\tilde{\mu}K_L+\tilde{\phi})=\tilde{\mu}b_L+H^{(\textrm{inv})}(\tilde{\phi}),$ where $H^{(\textrm{inv})}(\tilde{\phi}) \leq \min\left\{ b-\rho_L, b-\sigma_L \right\}$.

\begin{theorem}
The resilience of BFR-LRC is upper bounded by $\rho \leq b-H^{(\textrm{inv})}(\Mc)$.
\end{theorem}
When URA codes are used as local codes, we have the file size bound for resilience optimal codes as $\Mc \leq H(b-\rho) = \mu K_L + H(\phi)$, where $\mu=\floorb{\frac{b-\rho}{b_L}}$ and $\phi=b-\rho-\mu b_L$. Note that if $\phi \geq \min\left\{b_L-\rho_L, b_L-\sigma_L\right\}$, then $H(\sum_{i=1}^{\phi}b_i) = K_L$ since from any such $\phi$ blocks, one can regenerate all symbols or retrieve the content stored in the corresponding local group. Therefore, the case of having $\phi \geq \min\{b_L-\rho_L,$  $b_L-\sigma_L\}$ results in $\Mc=(\mu+1)K_L$ and we will mainly focus on the otherwise in the following.

\subsubsection{Local Regeneration with BFR-MSR Codes}

At first, we analyze the case where BFR-MSR codes are used inside local groups to have better trade-off in terms of repair bandwidth. When BFR-MSR codes are used, dimension of local code is given by
\begin{equation}
K_L = \sum_{i=1}^{b_L}H(b_i|b_1,\dots,b_{i-1})=k_L\alpha.
\end{equation} 

Since BFR-MSR is a class of URA codes, we can upper bound the resilience of BFR-MSR-LRC as 
$\rho \leq b - H^{(\textrm{inv})}(\Mc)$, where for BFR-MSR codes we have
\begin{equation}
H^{(\textrm{inv})}(\mu K_L+\phi)=\mu b_L+\varphi,
\end{equation}
for some $\mu \geq 0$ and $1 \leq \phi \leq K_L$ and $\varphi$ is determined from $\frac{K_L(\varphi-1)}{b_L-\rho_L} < \phi \leq \frac{K_L \varphi}{b_L-\rho_L}$. Then, we can derive the file size bound for optimal BFR-MSR-LRC as $\Mc \leq H(b-\rho)$, i.e.;

\begin{equation}
\Mc \leq \mu K_L + \frac{K_L \varphi}{b_L-\rho_L},
\end{equation} 
where $\mu=\floorb{\frac{b-\rho}{b_L}}$ and $\varphi=b-\rho-\mu b_L < \min\left\{b_L-\rho_L, b_L-\sigma_L\right\}$. If $\varphi \geq \min\left\{b_L-\rho_L, b_L-\sigma_L\right\}$, then $\phi=K_L$ and $\Mc \leq K_L(\mu+1)$.

\subsubsection{Local Regeneration with BFR-MBR Codes}

In the following, we focus on the case where the local groups form BFR-MBR codes. Then, the dimension of the local code is given by

\begin{equation}
\begin{split}
K_L & = \sum_{i=1}^{b_L}H(b_i|b_1,\dots,b_{i-1}) \\ &= 
\begin{cases}
\beta(k_Ld_L - \frac{k_L^2(b_L-\rho_L-1)}{2(b_L-\rho_L)}), & \textrm{if } \frac{d_L}{b_L-\sigma_L} \geq \frac{k_L}{b_L -\rho_L} \textrm{ and } \sigma_L \leq \rho_L \\
\beta(k_Ld_L - \frac{k_L^2(b_L-\sigma_L)(b_L+\sigma_L-2\rho_L-1)}{2(b_L-\rho_L)^2}), & \textrm{if } \frac{d_L}{b_L-\sigma_L} \geq \frac{k_L}{b_L -\rho_L} \textrm{ and } \sigma_L > \rho_L \\
\beta(\frac{k_Ld_L(\rho_L-\sigma_L+1)}{b_L-\sigma_L} + \frac{d_L^2(b_L-\rho_L)(b_L-\rho_L-1)}{2(b_L-\sigma_L)^2}), & \textrm{if } \frac{d_L}{b_L-\sigma_L} < \frac{k_L}{b_L -\rho_L} \textrm{ and } \sigma_L < \rho_L \\
\end{cases}
\end{split}
\end{equation}
Using the fact that BFR-MBR is URA code, the upper bound on the resilience of BFR-MBR-LRC is given by $\rho \leq b - H^{(\textrm{inv})}(\Mc)$, where for BFR-MBR codes we have
\begin{equation}
H^{(\textrm{inv})}(\mu K_L+\phi)=\mu b_L+\varphi,
\end{equation}
for some $\mu \geq 0$ and $1 \leq \phi \leq K_L$ and $\varphi$ is determined from
\begin{equation}
\varphi =
\begin{cases}
\beta(\frac{k_L d_L(\varphi-1)}{b_L-\rho_L} - \frac{k_L^2(\varphi-2)(\varphi-1)}{2(b_L-\rho_L)^2}) < \phi \leq \beta(\frac{k_L d_L \varphi}{b_L-\rho_L} - \frac{k_L^2\varphi(\varphi-1)}{2(b_L-\rho_L)^2}), \qquad \qquad \qquad  \textrm{if } \frac{d_L}{b_L-\sigma_L} \geq \frac{k_L}{b_L -\rho_L} \\
\frac{d_L(\varphi-1)\beta}{b_L-\sigma_L}(\frac{d_L}{2}+\frac{(2k_c-d_r)(b_L-\sigma_L-\varphi+2)}{2})) < \phi \leq \frac{d_L\varphi\beta}{b_L-\sigma_L}(\frac{d_L}{2}+\frac{(2k_c-d_r)(b_L-\sigma_L-\varphi+1)}{2})),  \quad \textrm{o.w}
\end{cases}
\end{equation}
where $k_c=\frac{k_L}{b-L-\rho_L}$ and $d_r=\frac{d_L}{b_L-\sigma_L}$. Now, the file size bound for an optimal BFR-MBR-LRC is given by
\begin{equation}
\Mc \leq \mu K_L + \Delta,
\end{equation}
where $\mu=\floorb{\frac{b-\rho}{b_L}}$, $\varphi=b-\rho-\mu b_L < \min\left\{b_L-\rho_L, b_L-\sigma_L\right\}$ and
\begin{equation}
\Delta = 
\begin{cases}
\beta(\frac{k_L d_L \varphi}{b_L-\rho_L} - \frac{k_L^2\varphi(\varphi-1)}{2(b_L-\rho_L)^2}) , &  \textrm{if } \frac{d_L}{b_L-\sigma_L} \geq \frac{k_L}{b_L -\rho_L} \\
\beta(\frac{k_Ld_L \varphi}{b_L-\rho_L}(\frac{b_L-\sigma_L-\varphi+1}{b_L-\sigma_L})+\frac{d_L^2\varphi (\varphi-1)}{2(b_L-\sigma_L)^2}), & o.w
\end{cases}
\end{equation}

If $\varphi \geq \min\left\{b_L-\rho_L, b_L-\sigma_L\right\}$, then $\Mc \leq K_L(\mu+1)$.

\subsubsection{Construction of BFR-MSR/MBR-LRC}




In order to construct codes for BFR-MSR/MBR-LRC, we utilize our earlier construction that is based on Duplicated Combination Block Designs.

\noindent \textbf{\textit{Construction VI:}} Consider a file $\fv$ of size $\Mc$.
\begin{itemize}
\item First, encode $\fv$ using $[N, K=\Mc, D=N-\Mc+1]$ Gabidulin code, $\Cc^{\rm{Gab}}$. The resulting codeword $\cv \in \Cc^{\rm{Gab}}$ is divided into $\frac{b}{b_L}$ disjoint local groups of size $\frac{Nb_L}{b}$.
\item Apply \textit{Construction III} to each local group. 
\end{itemize}    

\begin{remark}
\textit{Construction VI} requires $\frac{N}{b(b_L-\sigma_L)}$ to be integer since in each local group we utilize \textit{Construction III}, where a sub-code (regenerating code) is used with $\tilde{k}=\tilde{\Mc}$ and $\tilde{\Mc}$ is obtained by partitioning local file into $(b_L-\sigma_L)b_L$ parts. Furthermore, due to use of \textit{Construction III}, Construction VI results in codes with $\rho_L=0$. (This implies that code can tolerate any $\rho$ block erasures, but the content of the local group will be reduced even after a single block erasure in the local group.)
\end{remark}

\begin{corollary}
Construction VI yields an optimal code, with respect to resilience, when $b(b_L-\sigma_L) \mid N$ and $K_L \mid \Mc$.
\end{corollary}
\begin{proof}
We need to show that $E=b - \ceilb{\frac{M b_L}{K_L}}$ number of block erasures are tolerated since $\rho_L=0$. If $K_L \mid \Mc$, then $\Mc = \alpha_1 K_L$ and $E=b-\alpha_1 b_L$, which implies that the worst case erasure scenario is having $\frac{N}{K_L} - \alpha_1$ local groups with $b_L$ erasures. Then, total rank erasure is $K_L(\frac{N}{K_L}-\alpha_1)=N-K_L\alpha_1$, which can be corrected by Gabidulin code since $D-1=N-\Mc=N-\alpha_1 K_L$. 
\end{proof}

We will utilize \emph{Construction VI} to obtain optimal BFR-MSR-LRC. \emph{Construction VI} can be used to construct BFR-MBR-LRC as well, if BFR-MBR codes are used inside local groups.

\begin{corollary}
\textit{Construction VI} yields an optimal LRC with respect to file size bounds when $\frac{N}{b} \mid \Mc$ for BFR-MSR-LRC.
\label{cor:cons3}
\end{corollary}

\begin{proof}
If $\frac{N}{b} \mid \Mc$, then $\frac{K_L}{b_L-\rho_L} \mid \Mc$ (since $K_L = \frac{N b_L}{b}$) for $\rho_L=0$, which is the case. Therefore, we can write $\Mc = \frac{K_L}{b_L}(\alpha_1b_L + \beta_1)$ for $0 \leq \alpha_1 \leq \frac{b}{b_L}$ and $0 \leq \beta_1 \leq b_L-\sigma_L$. 
\begin{itemize}
\item If $\beta_1=0$, then $b-\rho=\ceilb{\frac{\Mc b_L}{K_L}}=\alpha_1 b_L$. Therefore $\mu=\floorb{\frac{b-\rho}{b_L}}=\alpha_1$ and $\varphi = b-\rho-\mu b_L=0$, which implies that $\mu K_L + \frac{K_L \varphi}{b_L} = \alpha_1 K_L = \Mc$.
\item If $\beta_1 \not = 0$, then $b-\rho=\ceilb{\frac{\Mc b_L}{K_L}}=\alpha_1 b_L + \beta_1$. Therefore $\mu=\floorb{\frac{b-\rho}{b_L}}=\alpha_1$ and $\varphi = b-\rho-\mu b_L = \beta_1$, which means $\mu K_L + \frac{K_L \varphi}{b_L} = \alpha_1 K_L + \frac{K_L \beta_1}{b_L}= \Mc$. 
\end{itemize}
\end{proof}


\section{Discussion}
\label{sec:Disc}

\subsection{Repair Delay}

In this section, we analyze the repair delay by considering the available bandwidth between the blocks. (In DSS architectures racks are connected through switches, referred to as top-of-rack (TOR) switches, and the model here corresponds to the communication delay between these switches.) Since in schemes that use BFR, a failed node requires $d_r\beta_{\textrm{BFR}}$ amount of data from each of $b-\sigma$ blocks, the repair delay (normalized with file size $\Mc$) can be calculated as follows.
\begin{equation}
RT_{\textrm{BFR}}=\max_{i \in \Bc_h} \frac{d_r\beta_{\textrm{BFR}}}{\Mc BW_i}, \quad \Bc_h = \left\{1,\dots,b-\sigma \right\},
\end{equation}
where $BW_i$ is the bandwidth for block $i$ and $\Bc_h$ is the set of blocks that help in the repairs. We assume the repairs through blocks are performed in parallel. Throughout this section, we assume that all bandwidths are identical ($BW_i=BW, \forall i$), hence the repair time of a failed node is given by $\frac{d_r\beta_{\textrm{BFR}}}{\Mc BW}$. 

In an identical setting, we can also analyze the repair delay of regenerating codes. Note that regenerating codes do not require any symmetric distribution of helper nodes among blocks. Hence repair delay of regenerating codes is,   
\begin{equation}
RT_{\textrm{RC}}(s) = \max_{i} \frac{d_i\beta_{\textrm{RC}}}{\Mc BW}, \quad s \in \Sc = \left\{ \left\{ d_1, \dots, d_{|\Bc|} \right\} \textrm{ s.t. } \sum_{i} d_i=d, \Bc = \left\{1,\dots,b \right\} \right\}
\end{equation}
where a system $s$ refers to a selection of $d_i$ (the number of helper nodes in block $i$) such that the sum of helper nodes is equal to $d$. In other words, repair delay will be affected by the block with the most helper node. Furthermore, the average repair delay can be calculated as
\begin{equation}
\textrm{E}[RT_{\textrm{RC}}(s)] = \frac{1}{|\Sc|} \sum_{s \in \Sc} RT_{\textrm{RC}.}(s)
\label{eq:avg_rep_time}
\end{equation} 

We may encounter some $d$ values, which are not attainable in the BFR model since for BFR codes we require $d \leq n-c$ because of the assumption that a failed node does not connect to any nodes in the same block. Furthermore, to compare the codes in a fair manner, we assume that regenerating codes connect to any $d$ nodes from $b-\sigma$ blocks. Henceforth, in our comparisons, we instead calculate repair delay for regenerating codes as 
\begin{equation}
RT_{\textrm{RC}}(s) = \max_{i} \frac{d_i\beta_{\textrm{RC}}}{\Mc BW}, \quad s \in \Sc = \left\{ \left\{ d_1, \dots, d_{|\Bc_h|} \right\} \textrm{ s.t. } \sum_{i} d_i=d, \Bc_h = \left\{1,\dots,b-\sigma \right\} \right\}
\end{equation}
where the helper nodes are chosen from $b-\sigma$ blocks but the number of helpers in each block is not necessarily the same, only requirement is to have the sum of the helpers equal to $d$. Average repair time can be calculated similar to \eqref{eq:avg_rep_time} where the difference being the system $s$ satisfies the condition that helpers are chosen from $b-\sigma$ blocks.

One can allow symmetric distribution of regenerating codes among $d-\sigma$ blocks as well, similar to BFR codes. We'll denote these codes by MSR-SYM or MBR-SYM in the following. 

We examine the case where $b=7$ and $n=21$, which means each block contains $c=3$ nodes. We also set $\sigma=3$ so that $b-\sigma=4$. Note that since we are comparing relative values, the values we assign to BW and $\Mc$ does not change the result as long as they are same across all comparisons. In other words, one can think our results as normalized repair delays. At first, we find out all possible $\rho, d_r, k_c, d$ and $k$ values accordingly from Section~\ref{sec:FileSize}. After identifying all possible parameter sets, we calculate repair times for both BFR-MSR and BFR-MBR codes. For regenerating codes, we choose all possible parameter sets as long as $d \leq n-\sigma c$ since we impose regenerating codes to perform repairs from $b-\sigma$ blocks. For each parameter set, we calculate repair times of all possible helper node allocations and then calculate the mean of those. The mean is reported for each data point in Fig.~\ref{fig:repair-delay} for one set of parameters for both MSR and MBR codes. Finally, we allow symmetric distribution of helper nodes among $b-\sigma$ blocks, and report MSR-SYM and MBR-SYM. 

\begin{figure}
\centering
\begin{subfigure}{.5\textwidth}
  \centering
  \includegraphics[width=0.8\linewidth]{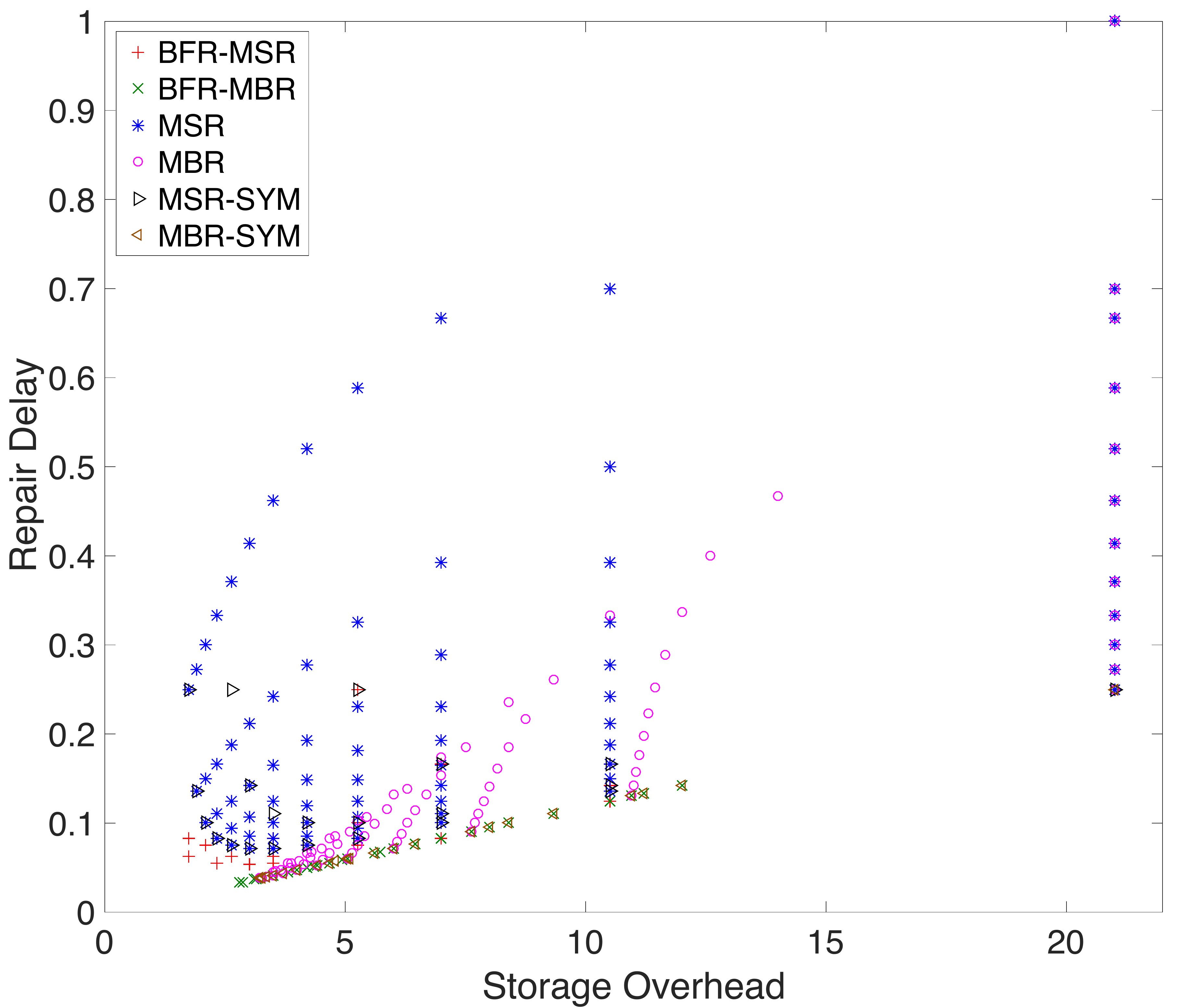}
  \caption{}
  \label{fig:repair-SO}
\end{subfigure}%
\begin{subfigure}{.5\textwidth}
  \centering
  \includegraphics[width=0.8\linewidth]{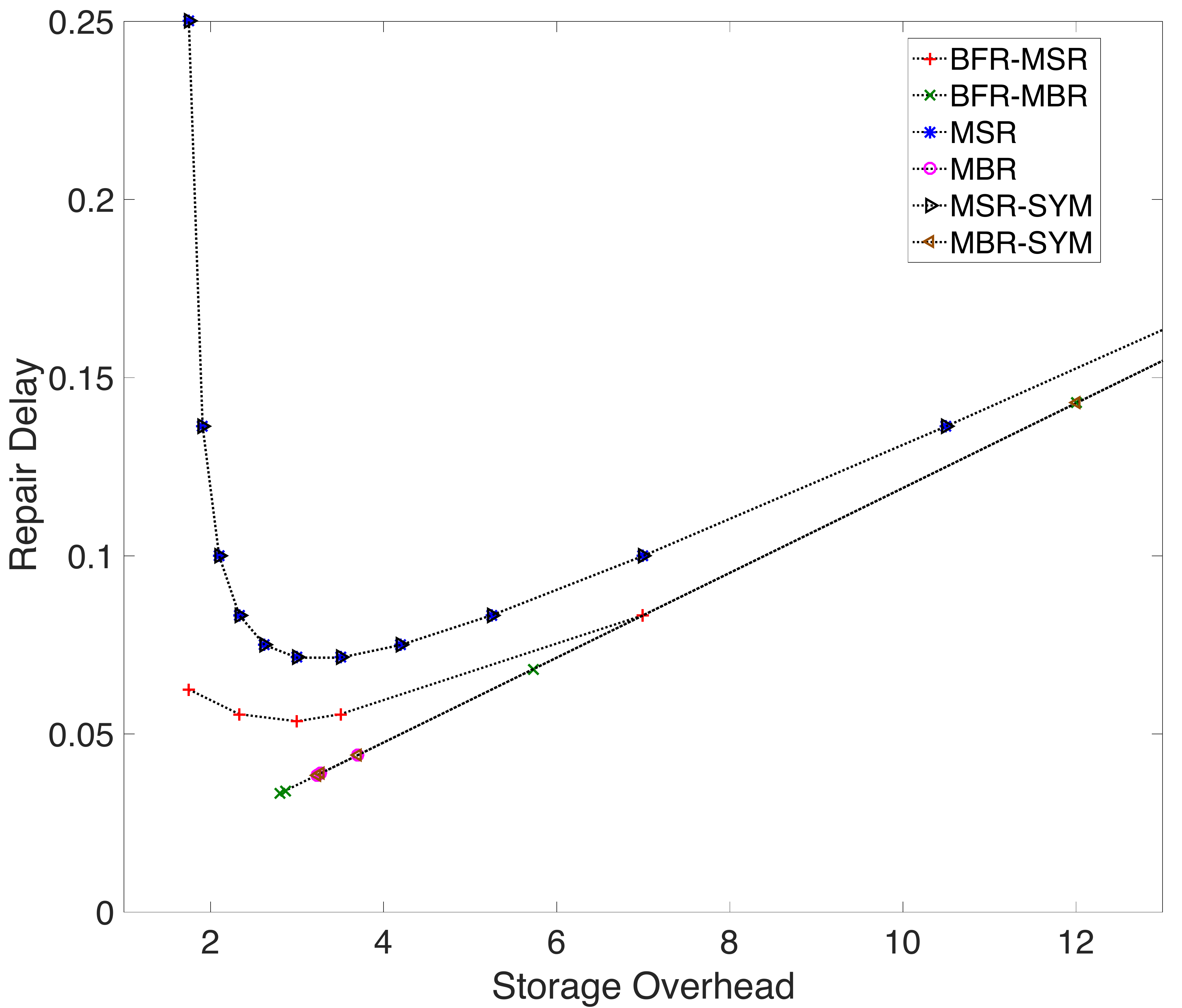}
  \caption{}
    \label{fig:envelope_repair-SO}
\end{subfigure}
\caption{Repair delay vs. storage overhead comparisons for $b=7$, $n=21$ and $\sigma=3$. (a) Data points for possible cases of each node. (b) Lower envelope of Fig.~\ref{fig:repair-SO} when zoomed in.}
\label{fig:repair-delay}
\vspace{-0.3in}
\end{figure}

Repair delay vs. storage overhead results are depicted on Fig.~\ref{fig:repair-delay}. In Fig.~\ref{fig:repair-SO} we indicate all data points, whereas in Fig.~\ref{fig:envelope_repair-SO}, lower envelope of Fig.~\ref{fig:repair-SO} for storage overheads less than 13. As expected, one can observe that there are more data points for regenerating codes than BFR, since BFR requires $b-\sigma \mid d$, which limits the number of possible sets of parameters for BFR. First, lower envelopes of MSR and MSR-SYM are the best repair times for MSR which are achieved when $d$ gets its highest possible value, $(b-\sigma)c$. In other words, all nodes are utilized hence there is no different possible connection schemes for MSR. Interestingly, given storage overhead, we observe that in some cases MSR codes perform better than MBR codes eventhough MBR codes minimizes repair bandwidth. On the other hand, when distributed symmetrically across blocks, MBR-SYM outperforms both MBR and MSR-SYM in all cases. When we compare BFR-MSR and BFR-MBR, we can observe that BFR-MBR has lower repair delay for all cases but still they perform the same when storage overhead gets larger. Furthermore, we observe that unlike MBR-MSR comparison, BFR performs more regularly, meaning BFR-MBR is better than BFR-MSR always. Next, if we compare all schemes, we observe that convex hulls of BFR-MBR, MBR and MBR-SYM follows the same line. Note that repair delay of BFR-MSR is below MSR and it performs the same as storage overhead increases. Also, BFR-MBR outperforms MSR-SYM and performs identical to MBR-SYM. For lower storage overheads, we observe that BFR codes operate well (both BFR-MSR and BFR-MBR) whereas existing RC or RC-SYM codes (MBR and MBR-SYM codes do not even exist for overhead below 3.23) performs worse than BFR. 

Finally, note that within BFR schemes, we may encounter different $\alpha$ and $\beta$ values depending on the parameters. Differentiating between cases, let BFR1 denote the schemes with $d_r \geq k_c$ and $\rho \geq \sigma$, BFR2 denote $d_r \geq k_c$ and $\rho < \sigma$, and BFR3 denote $d_r < k_c$ and $\rho \geq \sigma$. The resulting $\beta$ values for these schemes may not be the same for same $k$ and $d$. In Fig.~\ref{fig:BFR-cases}, we examine these BFR schemes (for storage overhead less than $10$). We observe that the convex hull for BFR-MBR is a line and all BFR-MBR schemes operate on that line. Furthermore,different MSR schemes can perform better depending on the storage overhead. 


\begin{figure}
\centering
\begin{subfigure}{.5\textwidth}
  \centering
  \includegraphics[width=0.8\linewidth]{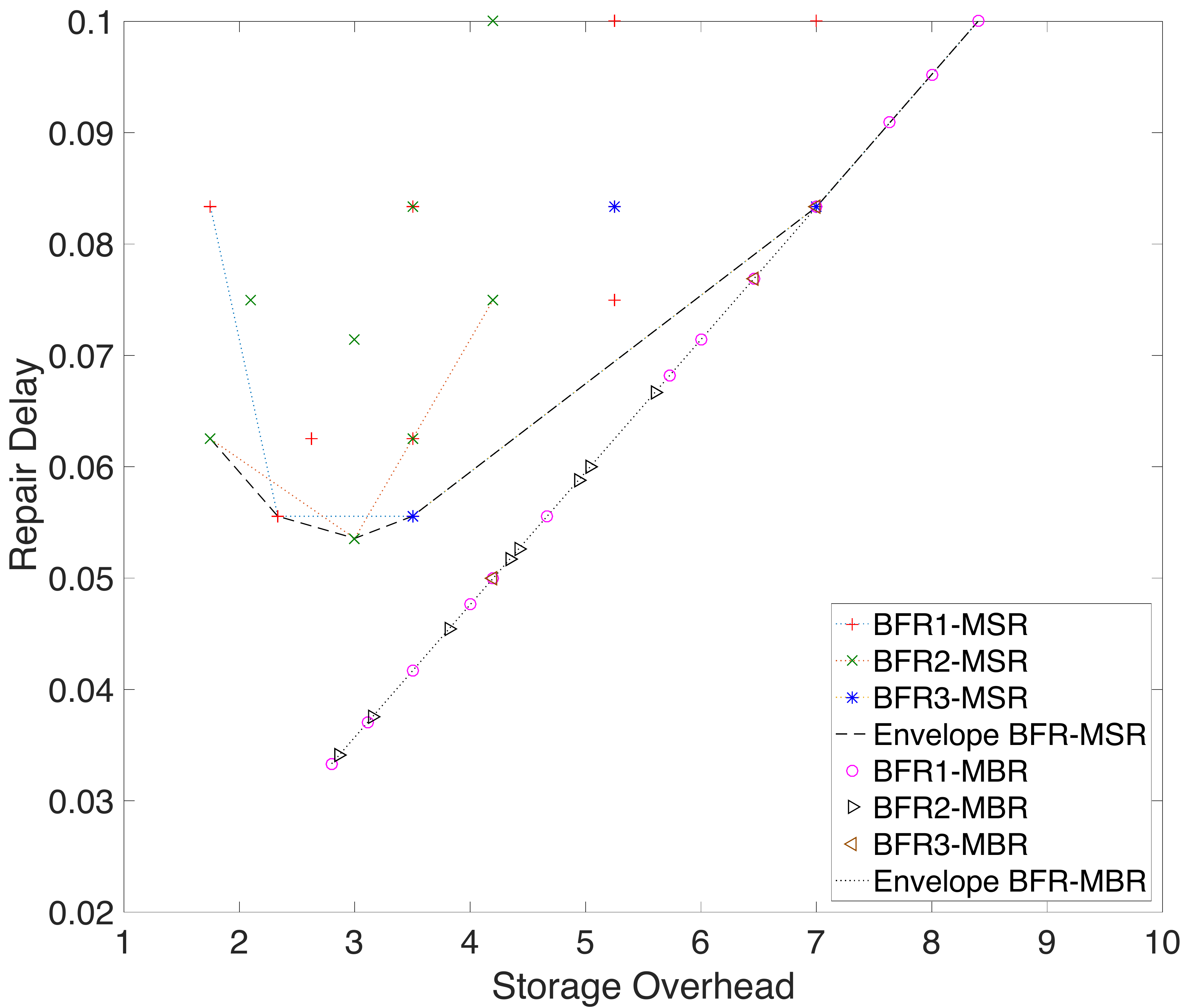}
  \caption{}
\end{subfigure}%
\begin{subfigure}{.5\textwidth}
  \centering
  \includegraphics[width=0.8\linewidth]{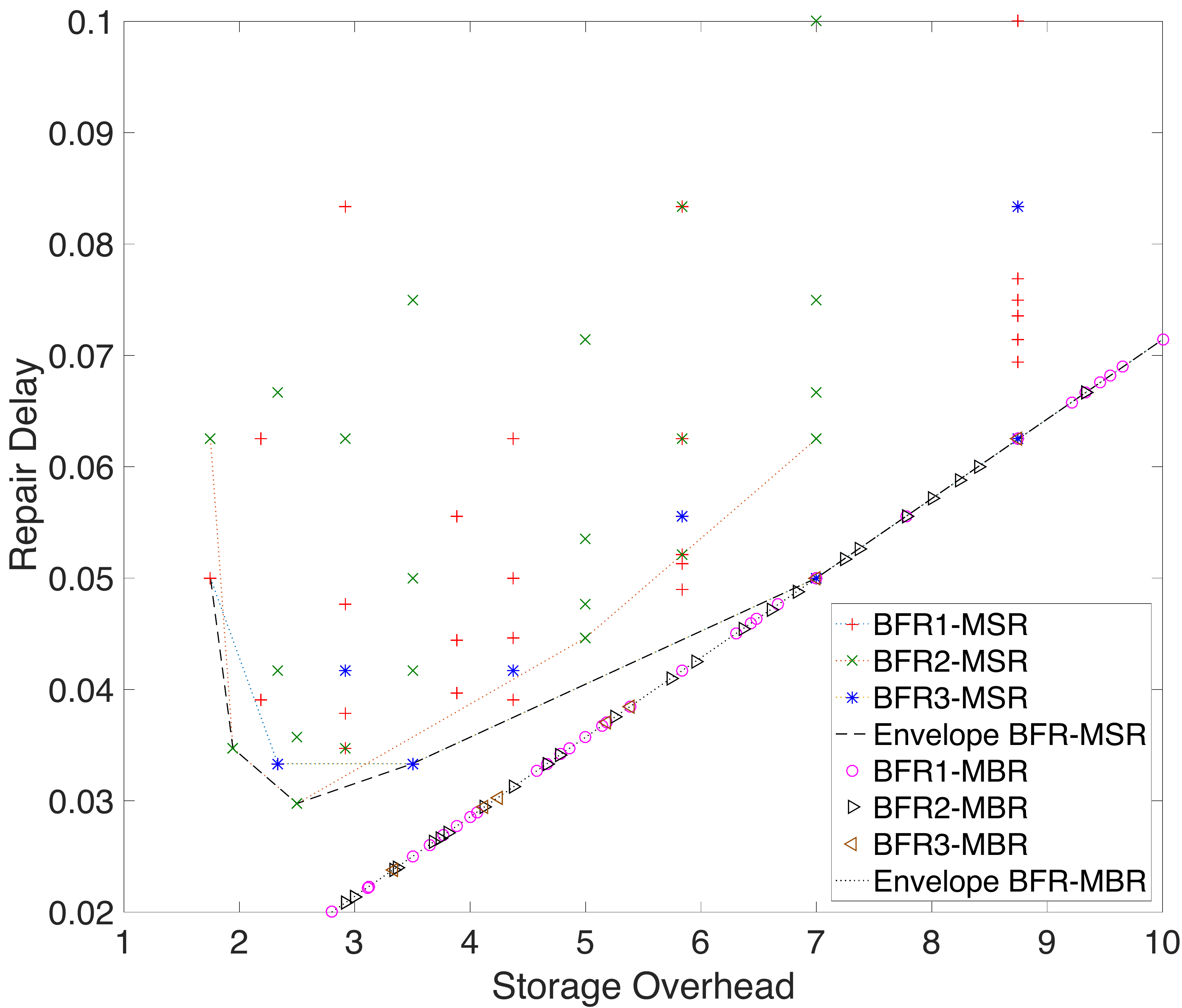}
  \caption{}
\end{subfigure}
\caption{Repair delay vs. storage overhead comparisons for $b=7$ and $\sigma=3$, (a) when $n=21$, (b) when $n=35$.}
\vspace{-0.3in}
\label{fig:BFR-cases}
\end{figure}

\subsection{Relaxed BFR}

It's challenging to find general constructions that allow for repair and data collection schemes to operate with \textit{any} $\rho$ and $\sigma$ blocks with \textit{any} subset of $d_r$ or $k_c$. Hence, in this section, we relax ``\textit{any}" requirement of BFR to table-based recovery schemes that guarantee the existence of helper/data recovery nodes for every possible repair/data collection scenarios. By altering this, we are able to use another combinatorial design and obtain relaxed BFR codes (R-BFR) for wider set of parameters. \footnote{We note that under this setting the fundamental limits will be different than BFR codes. The focus here is more on operating at BFR performance under a relaxation of ``any" requirement.}

\subsubsection{Resolvable Balanced Incomplete Block Design (RBIBD)}

\begin{definition}
A parallel class is the set of blocks that partition the point set. A resolvable balanced incomplete block design is a $(v,\kappa,\lambda)-BIBD$ whose blocks can be partitioned into parallel classes. 
\end{definition}

An example of $(9,3,1)-RBIBD$ is given below, where each column (consisting of 9 points in 3 blocks) forms a parallel class.
\begin{equation}
\begin{split}
\left\{1,2,3\right\} \hspace{0.5 cm} \left\{1,4,7\right\} \hspace{0.5 cm} \left\{1,5,9\right\} \hspace{0.5 cm} \left\{1,6,8\right\} \\
\left\{4,5,6\right\} \hspace{0.5 cm} \left\{2,5,8\right\} \hspace{0.5 cm} \left\{2,6,7\right\} \hspace{0.5 cm} \left\{2,4,9\right\} \\
\left\{7,8,9\right\} \hspace{0.5 cm} \left\{3,6,9\right\} \hspace{0.5 cm} \left\{3,4,8\right\} \hspace{0.5 cm} \left\{3,5,7\right\} \\
\end{split}
\label{eq:assignment-RBIBD}
\end{equation}


\subsubsection{R-BFR-RC with RBIBD}

In this section, we show that RBIBDs can be used to construct R-BFR codes for any $\rho \geq 0$ and $\sigma \geq 1$. Considering RBIBDs (with $\lambda=1$) as defined above, we construct blocks each containing the same number of nodes that store symbols belonging to different partitions. For instance, utilizing \eqref{eq:assignment-RBIBD}, a block can be formed to contain 12 nodes, where 4 nodes store symbols in the form of $\left\{1,2,3\right\}$, (referred to as ``type" below), 4 nodes store symbols of the type $\left\{4,5,6\right\}$ and the other 4 nodes store symbols of the type $\left\{7,8,9\right\}$. We refer to blocks of the same type as sub-block. Assume that one of the nodes that store the symbols of the type $\left\{1,2,3\right\}$ is failed. A newcomer may download symbols from any other subset of blocks but instead of connecting to any $d_r$ nodes from each block, we consider connecting to any $\frac{d_r}{3}$ nodes of type $\left\{1,4,7\right\}$, any $\frac{d_r}{3}$ nodes of type $\left\{2,5,8\right\}$ and any $\frac{d_r}{3}$ nodes of type $\left\{3,6,9\right\}$ in the second block and so on. \footnote{This construction can be viewed as a generalization of table-based repair for regenerating codes, see e.g., fractional regenerating codes proposed in \cite{Rouayheb:Fractional10}.} Similarly, DC can connect to any $\frac{k_c}{3}$ from each sub-blocks. Therefore, the requirement of \textit{any} set of nodes from a block is changed to \textit{any} subset of nodes from a \emph{sub-block}. Note that, RBIBD still preserves any $\rho$ and $\sigma$ properties.

In the general case of any $\rho$ and $\sigma$, we still have the same relationship as before, c.f., \eqref{eq:assignment-RBIBD_2}, since repair property with any $b-\sigma$ blocks or DC property with any $b-\rho$ blocks does not change this relationship but instead it only alters which $d_r$ and $k_c$ nodes that are connected in a block, where  $d_r=\frac{d}{b-\sigma}$, $k_c=\frac{k}{b-\rho}$. Also, to ensure the repair and DC properties, $c =\frac{n}{b}\geq \max\{k_c, d_r\}$ must be satisfied.
\begin{equation}
M=v\tilde{M}, k=\frac{v}{\kappa}\tilde{k}, d=\kappa\tilde{d}, \alpha=\kappa\tilde{\alpha}, \beta=\tilde{\beta}. 
\label{eq:assignment-RBIBD_2}
\end{equation}  

With this construction, the same steps provided for BFR-MSR and BFR-MBR cases in the previous section can be followed. In the general case, we have three cases depending on the values of $d_r$, $k_c$, $\sigma$ and $\rho$ and the corresponding cut values can be found using Theorems~\ref{thm:case1}, \ref{thm:case3}, and Lemma \ref{thm:case2}. Solving for these minimum storage and bandwidth points, optimal $\tilde{k}$ values can be found. We provide these result in the following subsections.

\subsubsection{R-BFR-MSR}

To construct a R-BFR-MSR code, we set each sub-code $\tilde{\Cc}$ as an MSR code which has 
$\tilde{\alpha}=\frac{\tilde{\Mc}}{\tilde{k}}$ and  $\tilde{d}\tilde{\beta}=\frac{\tilde{\Mc}\tilde{d}}{\tilde{k}(\tilde{d}-\tilde{k}+1)}$.
This together with \eqref{eq:assignment-RBIBD_2} results in the following parameters for R-BFR-MSR construction
\begin{equation}
\alpha=\tilde{\alpha}\kappa=\frac{\Mc}{k} \quad d\beta=\tilde{d}\kappa\tilde{\beta}=\frac{\Mc d}{k(d-\frac{k\kappa^2}{v}+\kappa)}.
\label{eq:rbibd-bfr-sr}
\end{equation}

Using the relationships above together with \eqref{eq:case1_MSR}, \eqref{eq:case2_MSR} and \eqref{eq:case3_MSR} we obtain the optimal value as follows: (note that $b=\frac{v-1}{\kappa-1}$)
\begin{equation}
\tilde{k} = 
\begin{cases} 
\frac{\kappa^2(b-\rho)}{(\kappa^2-v)(b-\rho) + v}, & d_r \geq k_c \textrm{ and } \sigma \leq \rho \\
\frac{\kappa^2(b-\rho)}{\kappa^2(b-\rho) - v(b-\sigma)}, & d_r \geq k_c \textrm{ and } \sigma < \rho \\
\frac{d(b-\rho-1) + \kappa(b-\sigma)}{\kappa(b-\sigma)}, & d_r < k_c \textrm{ and } \sigma \leq \rho.
\end{cases}
\end{equation}

\subsubsection{R-BFR-MBR}
To construct R-BFR-MBR code, we set each sub-code $\tilde{\Cc}$ as an MBR code which has 
$\tilde{\alpha}=\tilde{d}\tilde{\beta}=\frac{2\tilde{\Mc}\tilde{d}}{\tilde{k}(2\tilde{d}-\tilde{k}+1)}$. This together with \eqref{eq:assignment-RBIBD_2} results in the following parameters for our R-BFR-MBR construction
\begin{equation}
\alpha=d\beta=\frac{2\Mc d}{k(2d-\frac{k\kappa^2}{v}+\kappa)}.
\label{eq:rbibd-bfr-br}
\end{equation}

We can solve for optimal value in MBR case using the equations above together with \eqref{eq:case1_MBR}, \eqref{eq:case2_MBR} and \eqref{eq:case3_MBR}. Resulting optimal values are as follows (note that $b=\frac{v-1}{\kappa-1}$)
\begin{equation}
\tilde{k} = 
\begin{cases} 
\frac{\kappa^2(b-\rho)}{(\kappa^2-v)(b-\rho) + v}, & d_r \geq k_c \textrm{ and } \sigma \leq \rho \\
\frac{\kappa^2(b-\rho)^2}{\kappa^2(b-\rho)^2 - v(b-\sigma)(b+\sigma-2\rho-1)}, & d_r \geq k_c \textrm{ and } \sigma < \rho \\
\frac{\frac{2d(b-\rho-1)}{\kappa}-b+\sigma\pm\sqrt{(\frac{2d(b-\rho+1)}{\kappa}-b+\sigma)^2-\frac{4d^2(b-\rho)(b-\rho-1)}{v}}}{2(b-\sigma)}, & d_r < k_c \textrm{ and } \sigma \leq \rho.
\end{cases}
\end{equation}



\section{Conclusion}
\label{sec:Conclusion}
We introduced the framework of block failure resilient (BFR) codes that can recover data stored in the system from a subset of available blocks with a load balancing property. Repairability is then studied, file size bounds are derived, BFR-MSR and BFR-MBR points are characterized, explicit code constructions for a wide parameter settings are provided for limited range of $\sigma$ and $\rho$. We then analyzed BFR for a broader range of parameters and characterized the file size bounds in these settings and also proposed code constructions achieving the critical points on the trade-off curve. Locally repairable BFR codes are studied where the upper bound on the resilience of DSS is characterized and two-step encoding process is proposed to achieve optimality. We finally analyzed the repair delay of BFR codes and compare those with regenerating codes, and provide constructions with table-based repair and data recovery properties.

Constructions reported here are based on combinatorial designs, which necessitate certain parameter sets. As a future work, optimal code constructions for the BFR model can be studied further (especially for $\rho>0$ case). Also, the file size bound for the case of having $d_r \geq k_c$ and $\sigma > \rho$ is established here, and we conjecture that this expression corresponds to the min-cut. The proof for this conjecture resisted our efforts thus far. Furthermore, repair delay with a uniform BW assumption is studied here. Different bandwidth and more realistic communication schemes (e.g., queuing models) can be be studied. Further, system implementations can be performed and more realistic analysis can be made for disk storage and distributed  (cloud/P2P) storage scenarios.



\bibliographystyle{IEEEtran} 
\bibliography{IEEEabrv,./myref}

\begin{appendices}

\section{Concatenated Gabidulin and MDS coding}
\label{app:gab_mds}

Set $K=(b-\rho)k_c$, and consider data symbols $\{u_0,\dots,u_{K-1}\}$.
\begin{itemize}
	\item Use $[N=K+\rho k_c,K,D]$ Gabidulin code to encode $\{u_0,\dots,u_{K-1}\}$ to length-$N$ codeword $(x_1,\dots,x_N)$. That is $(x_1,\dots,x_N)=(f(g_1),\dots,f(g_N))$, where the linearized polynomial $f(g)=u_0 g^{[0]}+\dots+u_{K-1}g^{[K-1]}$
	is constructed with $N$ linearly independent, over
	$\FF_q$, generator elements $\{g_1,\dots,g_N\}$ each in $\FF_{q^M}$; 
	and its coefficients are selected by the length-$K$ input vector. We represent this operation by writing 
	$\xv=\uv\Gm_{\textrm{MRD}}$.
	\item Split resulting $N$ symbols $\{x_1,\dots,x_N\}$ into $b$ blocks each with $k_c$ symbols. We represent this operation by double indexing the codeword symbols, i.e., $x_{i,j}$ is the symbol at block $i$ and $j$ for $i=1,\dots,b$, $j=1,\dots,k_c$. We also denote the resulting sets with the vector notation, $\xv_{i,1:k_c}=(x_{i,1},x_{i,2},\dots,x_{i,k_c})$ for block $i$.
	\item Use an $[n=c,k=k_c,d]$ MDS array code for each block to construct additional parities. Representing the output symbols as $\yv_{i,1:c}$ we have $\yv_{i,1:c}=\xv_{i,1:k_c}G_{\textrm{MDS}}$ for each block $i$, where $G_{\textrm{MDS}}$ is the encoding matrix of the MDS code over $\FF_q$.
	For instance, if a systematic code is used, $\xv_{i,1:k_c}$ is encoded into the vector $\yv_{i,1:n}=(x_{i,1},\dots,x_{i,k_C},$ $p_{i,1},\dots,p_{i,c-k_c})$ for each block $i=1,\dots,b$.
\end{itemize}

In the resulting code above if one erases $\rho$ blocks and any $c-k_c$ symbols from the remaining blocks, the remaining $(b-\rho)k_c$ symbols form linearly independent evaluations of the underlying linearized polynomial which can be decoded due to to the Gabidulin code from which the data symbols can be recovered and hence, by re-encoding, the pre-erasure of the version of the system can be recovered.

\section{Proof of Theorem \ref{thm:BFR-LRC}}
\label{app:res_theo}
\begin{proof}

In order to get an upper bound on the resilience of an BFR-LRC, the definition of resilience given in \eqref{def:resilience} is utilized similar to proof in \cite{Gopalan:Locality12,Papailiopoulos:Locally12}.  We iteratively construct a set $\Bc^* \subset \Bc$ so that $H(\cv(\Bc^*))< \Mc$. The algorithm is presented in Fig.~\ref{alg:resilience}. Let $b_i$ and $h_i$ represent the number of blocks and entropy included at the end of the $i$-th iteration. We define the following:
\begin{figure}
  \begin{algorithmic}[1]
    \algrule
    \State Set $\Bc^*_0=\emptyset$ and $i=1$
    \While  {$H(\cv(\Bc^*_{i-1}))<\Mc$} 
   		\State Pick a coded local group $\cv_{i} \notin \Bc^*_{i-1}$ s.t. $|\Bc_i \setminus \Bc^*_{i-1}| \geq \rho_L$
   		\If {$H(\cv(\Bc^*_{i-1}),\cv(\Bc_i))<\Mc$}
   			\State set $\Bc^*_i=\Bc^*_{i-1} \cup \Bc_i$
   		\ElsIf{$H(\cv(\Bc^*_{i-1}),\cv(\Bc_i)) \geq \Mc$ and $\exists \Bc'_i$ s.t. $\Bc'_i= \argmax\limits_{\Bc'_i \subset \Bc_i}H(\cv(\Bc^*_{i-1}),\cv(\Bc'_i))<\Mc$}
   			\State set $\Bc^*_i=\Bc^*_{i-1} \cup \Bc'_i$
		\EndIf
		\State $i=i+1$
   	\EndWhile
   	\State Output: $\Bc^*=\Bc^*_{i-1}$
   	\algrule
  \end{algorithmic}
  \caption{Construction of a set $\Bc^*$ with $H(\cv(\Bc^*))<\Mc$ for BFR-LRC}
  \label{alg:resilience}
\end{figure}
\begin{eqnarray}
\label{eq:block_diff}
b_i & = &|\Bc^*_{i}|-|\Bc^*_{i-1}| \leq b_L,  \\
h_i & = & H(\cv(\Bc^*_{i}))-H(\cv(\Bc^*_{i-1})) \leq K_L.
\label{eq:entropy_diff}
\end{eqnarray}

Assume that the algorithm outputs at $(l+1)^{th}$ iteration then it follows from \eqref{eq:block_diff} and \eqref{eq:entropy_diff} that 
\begin{equation}
|\Bc^*|=|\Bc^*_l|=\sum_{i=1}^{l} b_i \quad\quad\quad H(\cv(\Bc^*))=H(\cv(\Bc^*_l))=\sum_{i=1}^{l} h_i.
\label{eq:l_block}
\end{equation}
The analysis of algorithm is divided into two cases as follows.
\begin{itemize}
\item Case 1: [Assume that the algorithm exits without ever entering line 7.] We have \begin{equation} 
\begin{split}
h_i & = H(\cv(\Bc^*_{i}))-H(\cv(\Bc^*_{i-1}))= H(\cv(\Bc^*_i \setminus \Bc^*_{i-1} )| \cv(B^*_{i-1})) \leq (b_i-\rho_L)\frac{K_L}{b_L-\rho_L}. \end{split}
\end{equation}
From which we can obtain $b_i \geq \frac{h_i(b_L-\rho_L)}{K_L}+\rho_L$. Then, we derive the following:
\begin{equation}
\begin{split}
|\Bc^*|=|\Bc^*_l| & = \sum_{i=1}^{l} b_i \geq \sum_{i=1}^{l} \left( \frac{h_i(b_L-\rho_L)}{K_L}+\rho_L \right) = \frac{b_L-\rho_L}{K_L} \sum_{i=1}^{l} h_i + \rho_L l. 
\end{split}
\label{eq:B^*}
\end{equation}

Similar to proof in \cite{Papailiopoulos:Locally12}, we have

\begin{equation}
\sum_{i=1}^{l} h_i = \ceilb{\frac{\Mc}{K_L/(b_L-\rho_L)}}\frac{K_L}{b_L-\rho_L}-\frac{K_L}{b_L-\rho_L}
\quad\quad
l = \ceilb{\frac{\Mc}{K_L}}-1.
\label{eq:l}
\end{equation}

By combining \eqref{eq:B^*} with the above, we obtain 

\begin{equation}
|\Bc^*|\geq \ceilb{\frac{\Mc(b_L-\rho_L)}{K_L}} -1 + \rho_L \left(\ceilb{\frac{\Mc}{K_L}}-1\right).
\label{eq:case1}
\end{equation}

\item Case 2: [The algorithm exits after entering line 7.] For this case, we have 
\begin{equation}
H(\cv(\Bc^*_{i-1}),\cv(\Bc_i)) \geq \Mc.
\end{equation} 
 
In each iteration step we add $K_L$ entropy, hence we have 
\begin{equation}
l \geq \ceilb{\frac{\Mc}{K_L}}.
\label{eq:l-2}
\end{equation}

For $i \leq l-1$, same as before, we have $b_i \geq \frac{h_i(b_L-\rho_L)}{K_L}+\rho_L$. For $i=l$, we have $b_l \geq \frac{h_l}{K_L/(b_L-\rho_L)}$. Hence, it follows from \eqref{eq:l_block} and \eqref{eq:l-2} that 
\begin{equation}
\begin{split}
|\Bc^*| & = \sum_{i=1}^{l} b_i \geq \sum_{i=1}^{l-1} \left(\frac{h_i(b_L-\rho_L)}{K_L}+\rho_L\right) + \frac{h_l}{K_L/(b_L-\rho_L)} = \frac{b_L-\rho_L}{K_L}\sum_{i=1}^{l}h_i + (l-1)\rho_L \\
& \geq \frac{b_L-\rho_L}{K_L} \left(\ceilb{\frac{\Mc(b_L-\rho_L)}{K_L}}\frac{K_L}{b_L-\rho_L} -\frac{K_L}{b_L-\rho_L}\right) + \left(\ceilb{\frac{\Mc}{K_L}}-1\right)\rho_L \\
& =\ceilb{\frac{\Mc(b_L-\rho_L)}{K_L}} -1 + \left(\ceilb{\frac{\Mc}{K_L}}-1\right)\rho_L.
\end{split}
\label{eq:case2}
\end{equation}

Therefore, by combining \eqref{def:resilience}, \eqref{eq:case1} and \eqref{eq:case2} we have 

\begin{equation}
\rho \leq b - \ceilb{\frac{\Mc(b_L-\rho_L)}{K_L}} - \left(\ceilb{\frac{\Mc}{K_L}}-1\right)\rho_L.
\end{equation}
\end{itemize}
\end{proof}

\end{appendices}


\end{document}